\documentclass[11pt,a4paper]{article}

\usepackage{amsthm,amsmath,amsfonts,amssymb}
\usepackage[margin=2cm]{geometry}
\usepackage{colonequals}
\usepackage{graphicx}
\usepackage{ifthen}

\usepackage{xcolor}
\usepackage{tikz}
\usetikzlibrary{decorations.pathreplacing}
\usetikzlibrary{positioning}
\usepackage{quantikz}
\usepackage{authblk}

\usepackage{algorithm}
\usepackage{float}
\usepackage{algpseudocode}

\usetikzlibrary{shapes.geometric}
\newboolean{ElectronicVersion}
\setboolean{ElectronicVersion}{true}

\ifthenelse{\boolean{ElectronicVersion}}{
    \usepackage[letterpaper=true,pdftex,bookmarks,pagebackref,
	plainpages=false, 
        pdfpagelabels=true 
        ]{hyperref}}{}

\makeatletter
\newtheorem*{rep@theorem}{\rep@title}
\newcommand{\newreptheorem}[2]{%
\newenvironment{rep#1}[1]{%
 \def\rep@title{#2 \ref*{##1}}%
 \begin{rep@theorem}}%
 {\end{rep@theorem}}}
\makeatother

\usepackage{hyperref}
\hypersetup{
    bookmarksnumbered=true, 
    unicode=false, 
    pdfstartview={FitH}, 
    pdftitle={A Quantum Time–Space Tradeoff for Directed $st$-Connectivity}, 
    pdfauthor={Stacey Jeffery and Galina Pass}, 
    pdfsubject={}, 
    pdfcreator={}, 
    pdfproducer={}, 
    pdfkeywords={}, 
    pdfnewwindow=true, 
    colorlinks=true, 
    linkcolor=blue, 
    citecolor=blue, 
    filecolor=blue, 
    urlcolor=blue 
}

\newcommand{\eq}[1]{\hyperref[eq:#1]{(\ref*{eq:#1})}}
\renewcommand{\sec}[1]{\hyperref[sec:#1]{Section~\ref*{sec:#1}}}
\newcommand{\thm}[1]{\hyperref[thm:#1]{Theorem~\ref*{thm:#1}}}
\newcommand{\lem}[1]{\hyperref[lem:#1]{Lemma~\ref*{lem:#1}}}
\newcommand{\cor}[1]{\hyperref[cor:#1]{Corollary~\ref*{cor:#1}}}
\newcommand{\app}[1]{\hyperref[app:#1]{Appendix~\ref*{app:#1}}}
\newcommand{\tab}[1]{\hyperref[tab:#1]{Table~\ref*{tab:#1}}}
\newcommand{\defin}[1]{\hyperref[def:#1]{Definition~\ref*{def:#1}}}
\newcommand{\fig}[1]{\hyperref[fig:#1]{Figure~\ref*{fig:#1}}}
\newcommand{\clm}[1]{\hyperref[claim:#1]{Claim~\ref*{claim:#1}}}
\newcommand{\conj}[1]{\hyperref[conj:#1]{Conjecture~\ref*{conj:#1}}}
\newcommand{\rem}[1]{\hyperref[rem:#1]{Remark~\ref*{rem:#1}}}
\newcommand{\algo}[1]{\hyperref[algo:#1]{Algorithm~\ref*{algo:#1}}}

\newcommand{\thmthm}[2]{\hyperref[thm:#1]{Theorem~\ref*{thm:#1}} and~\hyperref[thm:#2]{\ref*{thm:#2}}}
\newcommand{\lemlem}[2]{\hyperref[lem:#1]{Lemma~\ref*{lem:#1}} and~\hyperref[lem:#2]{\ref*{lem:#2}}}

\newtheorem{theorem}{Theorem}[section]
\newtheorem{lemma}[theorem]{Lemma}

\newtheorem{corollary}[theorem]{Corollary}

\newtheorem{claim}[theorem]{Claim}

\newtheorem{remark}[theorem]{Remark}
\newtheorem{definition}[theorem]{Definition}

\usepackage{color}
\definecolor{darkgreen}{rgb}{0,.5,0}
\definecolor{darkred}{rgb}{.7,.3,.3}
\definecolor{deepblue}{rgb}{0,.1,.7}

\def\ket#1{{\lvert}#1\rangle}
\def\bra#1{{\langle}#1\rvert}

\def\braket#1#2{{{\langle}#1\vert}#2\rangle}
\def\abs#1{\left| #1 \right|}

\def\norm#1{\left\| #1 \right\|}

\def\tO{{\widetilde{O}}}
\def\s{{\sf s}}
\def\t{{\sf t}}
\def\u{{\sf u}}
\def\v{{\sf v}}

\newcommand\restr[2]{{
  \left.\kern-\nulldelimiterspace 
  #1 
  \vphantom{\big|} 
  \right|_{#2} 
  }}

\title{A Quantum Time-Space Tradeoff for Directed $st$-Connectivity}
\author[1,2]{Stacey Jeffery}
\author[2]{Galina Pass}
\affil[1]{CWI, Amsterdam}
\affil[2]{QuSoft \& University of Amsterdam}
\date{}

\begin{document}
\vskip10pt

\maketitle

\begin{abstract}
Directed $st$-connectivity (\textsc{dstcon}) is the problem of deciding if there exists a directed path between a pair of distinguished vertices $s$ and $t$ in an input directed graph. This problem appears in many algorithmic applications, and is also a fundamental problem in complexity theory, due to its ${\sf NL}$-completeness. We show that for any $S\geq \log^2(n)$, there is a quantum algorithm for $\textsc{dstcon}$ using space $S$ and time $T\leq 2^{\frac{1}{2}\log(n)\log(n/S)+o(\log^2(n))}$, which is an (up to quadratic) improvement over the best classical algorithm for any $S=o(\sqrt{n})$. Of the $S$ total space used by our algorithm, only $O(\log^2(n))$ is quantum space -- the rest is classical. This effectively means that we can trade off classical space for quantum time. 
\end{abstract}

\section{Introduction}

In the directed $st$-connectivity problem (\textsc{dstcon}), the input is a directed graph on $n$ vertices with two distinguished vertices $s$ and $t$, and the goal is to decide if there is a directed path from $s$ to $t$. This fundamental problem underlies a wide range of applications in (e.g.) logistics, databases \cite{abiteboul1995,ullman1989}, compilers \cite{aho2006,nielson1999}, and model checking \cite{clarke1999,baier2008}.

In addition to its practical applications, \textsc{dstcon} plays a central role in space-bounded complexity theory. It is ${\sf NL}$-complete (under ${\sf NC}^1$ reductions), where ${\sf NL}$ is the class of problems decidable by nondeterministic logspace machines. Consequently, understanding its complexity has broad implications for space-bounded computation. For example, a classical (or quantum) algorithm for \textsc{dstcon} using $O(\log(n))$ space would show that ${\sf L}$, the class of problems solvable in $O(\log(n))$ space (or its quantum analogue) contains ${\sf NL}$ -- a major breakthrough in either case. 

Currently, the smallest space complexity of any classical or quantum algorithm for \textsc{dstcon} is $S=O(\log^2(n))$, achieved by Savitch's (classical) algorithm. However, Savitch's algorithm achieves this low space complexity at the expense of a large \emph{quasipolynomial} time complexity, using 
$$T\leq 2^{\log^2(n)+O(\log(n))}$$ 
steps of computation. In contrast, a simple breadth-first search (BFS) algorithm solves this problem in just $T=O(n^3)$ steps (in the adjacency-matrix model), but at the expense of a much larger $S=\tO(n)$ space requirement. The fundamental nature of this problem, as well as its many applications, motivates understanding the best possible \emph{tradeoff} between the time and space needed to solve it. 
Progress was made by Barnes, Buss, Ruzzo and Schieber~\cite{barnes1998dstconnST}, who gave a classical algorithm that runs in time 
\begin{equation}\label{eq:classical-TS}
    T \leq 2^{\log^2 (\frac{n}{S}) +O(\log n\log\log n)}
\end{equation} given any space $S\geq \log^2(n)$. Their approach combines a breadth-first search with a clever recursive algorithm. 

For quantum algorithms, the study of space-bounded complexity is even more well motivated, as quantum memories are expected to be limited in size for the foreseeable future. It is an important question in which memory regimes we can still achieve speedups over classical algorithms, and time-space tradeoffs are a key part of this picture. We can make such a tradeoff even more useful by distinguishing between the \emph{quantum} space and \emph{classical} (or total) space needed by the algorithm, as quantum space is the scarce resource. 

For \textsc{dstcon}, quantum speedups are known only at the two extreme regimes of space. In the high-space setting, Dürr, Heiligman, H{\o}yer, and Mhalla~\cite{durr2004graphProblems} gave an $\tO(n^{1.5})$-time, $\tO(n)$-space algorithm using quantum search to build a spanning tree. In the low-space setting, Ref.~\cite{jeffery_et_al:LIPIcs.STACS.2025.54} recently gave a quadratic quantum speedup over Savitch’s algorithm, running in time $T\leq 2^{\frac{1}{2}\log^2 n+O(\log n)}$ with $S=O(\log^2 n)$ space. Between these two extremes, however, no quantum improvements over classical tradeoffs were known.

For the \emph{undirected} variant (\textsc{ustcon}), the picture is much clearer, at least as far as quantum algorithms go: in both the adjacency-matrix and edge-list models of graph access\footnote{In this work, unless otherwise stated, we use the adjacency-matrix model, which assumes the input is given via queries to an adjacency matrix. This distinction is only significant in high-space regimes where the time complexity is polynomial.}, quantum algorithms achieve optimal time and space simultaneously~\cite{belovs2012stcon,apers_et_al:LIPIcs.ESA.2023.10}.

However, existing classical and quantum approaches face significant obstacles in the directed case. A random walk on a directed graph, starting from $s$, may fail to find $t$ even when $t$ is reachable from $s$, if the walk leads from $s$ to some part of the graph from which $t$ is not reachable. 
Quantum walks, which are powerful tools for studying undirected graphs, do not generalize to directed graphs. 
Moreover, known classical algorithms for \textsc{dstcon} such as the time-space tradeoff algorithm in~\cite{barnes1998dstconnST} 
cannot be directly quantized. This stems from the fact that these classical algorithms are not reversible, and making them reversible using standard methods increases the space complexity. These difficulties highlight why directed connectivity remains a challenging and subtle problem. They also show that progress requires algorithmic ideas that go beyond random or quantum walks, and straightforward adaptations of classical tradeoffs.

\paragraph{Our Contribution:} We present the first nontrivial quantum time-space tradeoff for \textsc{dstcon}, by designing a new quantum algorithm for this problem
that, for any space bound $S\geq \log^2(n)$, runs in time
$$T\leq 2^{\frac{1}{2}\log (n)\log(\frac{n}{S})+O(\log n\log\log n)}.$$ 
In particular, our result yields a quantum speedup over the best known classical algorithm (see \eqref{eq:classical-TS}) in the regime $S = o(n^{1/2})$. We also show that, of the $S$ space, the required \emph{quantum} space is always $O(\log^2(n))$, which makes this result much more applicable to quantum computers with a limited number of qubits. This effectively means that we can tradeoff classical space for quantum time. We formally state this result in \thm{quantum_tradeoff}.

\paragraph{Our Techniques:}
The classical algorithm of Barnes et al.~\cite{barnes1998dstconnST} achieves a time-space tradeoff for \textsc{dstcon} by combining breadth-first search with a recursive subroutine that decides, for any pair of vertices $u$ and $v$, if there is a directed path from $u$ to $v$ of length at most $L$. We call this problem $\textsc{Dist}_L$. The breadth-first search saves on space by not traversing the whole graph, but only those vertices at a distance from $s$ that is a multiple of $L$ (the space then decreases as $L$ increases). The subroutine is used to find these vertices in a manner that is more space efficient, but less time efficient, than BFS. 

The BFS portion of this classical algorithm could be sped up using quantum search techniques, but this saves at most a polynomial factor in the time (and nothing in the space), which is not very interesting in the small-space regime (where we get our quantum improvement) where such polynomial factors are hidden by the $O(\log(n))$ in the exponent. 
On the other hand, a direct quantization of their subroutine for $\textsc{Dist}_L$ fails, in large part because this subroutine is not reversible, and making it so would increase the space complexity. Moreover, the subroutine calls itself recursively, with significant depth of recursion. A quantum speedup of this subroutine would most likely have bounded error. Naively composing bounded-error quantum subroutines to depth $d$ results in $\log^d$ factors, which can be significant. Recent techniques for composing bounded error quantum algorithms without log-factor overhead~\cite{belovs2024transducers,belovs2024purifier} could reduce this overhead to $c^d$ for some constant, but this could still be significant, if $c>1$. Overcoming these limitation requires a fundamentally different approach.

To address this, we design a new quantum algorithm for $\textsc{Dist}_L$ that is based on a recursively constructed \emph{switching network}. A switching network (\defin{switching-network}) is an undirected graph with \emph{terminals} $u$ and $v$, whose edges are labeled by Boolean variables $\{x_1,\dots,x_m\}$, that can switch the edges ``on'' or ``off'', by their truth value.  
Such a network naturally defines a Boolean function $f:\{0,1\}^m\rightarrow\{0,1\}$, with $f(x)=1$ (we then say the switching network \emph{accepts} $x$) if and only if $u$ and $v$ are connected by a path consisting of edges whose labels are true under the assignment $x$. An example is given in \fig{example}.

\begin{figure}
\centering
\begin{tikzpicture}[>=Stealth, node distance=1.5cm,
  vertex/.style={circle, fill=black, inner sep=1.5pt},
  edge/.style={thick}]

\node[vertex,label=left:$u$] (u) {};
\node[vertex,right of=u] (a) {};
\node[vertex,above right of=a] (b) {};
\node[vertex,below right of=a] (c) {};
\node[vertex,right of=a, xshift=1cm,label=right:$v$] (v) {};

\path[edge] (u) edge node[above] {$x_1$} (a);
\path[edge] (a) edge node[above left] {$x_2$} (b);
\path[edge] (a) edge node[below left] {$x_3$} (c);
\path[edge] (b) edge node[above] {$x_4$} (v);
\path[edge] (c) edge node[below] {$x_5$} (v);

\node[draw=none, below=1.25cm of a, font=\small] {(L)};

\node[vertex,right=7cm of u,label=left:$u$] (u2) {};
\node[vertex,right of=u2] (a2) {};
\node[vertex,above right of=a2] (b2) {};
\node[vertex,below right of=a2] (c2) {};
\node[vertex,right of=a2, xshift=1cm,label=right:$v$] (v2) {};

\path[ultra thick] (u2) edge node[above] {$x_1$} (a2);
\path[dashed] (a2) edge node[above left] {$x_2$} (b2);
\path[ultra thick] (a2) edge node[below left] {$x_3$} (c2);
\path[dashed] (b2) edge node[above] {$x_4$} (v2);
\path[ultra thick] (c2) edge node[below] {$x_5$} (v2);

\node[draw=none, below=1.25cm of a2, font=\small] {(R)};

\end{tikzpicture}
\caption{(L) An example of a switching network. (R) The same switching network, with edges switched on (thick) or off (dashed) by the assignment $x=10101$. In this example, $u$ and $v$ are connected by a path of edges labelled by variables that are true under $x$, and so the switching network accepts $x$.}\label{fig:example}
\end{figure}
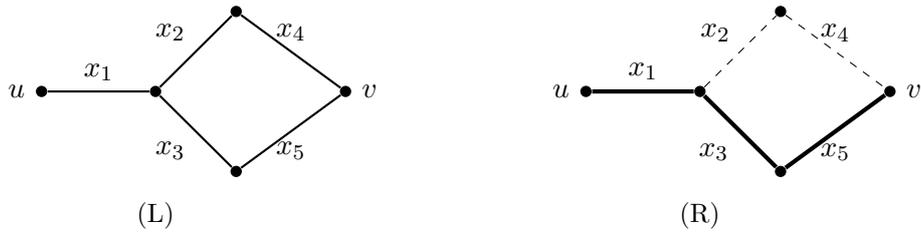

Switching networks are a natural and well-motivated model in classical computing, and have been used to study classical space-bounded complexity (see, e.g.,~\cite{potechin2015analyzing}).
They also give a simple way of designing quantum algorithms, as any switching network can be compiled into a quantum algorithm for the associated function, whose time and space complexity depend on certain properties of the switching network.
Quantum algorithms for evaluating switching networks were developed in~\cite{jeffery2017formula} using span program techniques inspired by~\cite{belovs2012stcon}, and a subsequent work~\cite{jarret2018connectivity} provided a tight analysis of these span program algorithms for arbitrary switching networks, in terms of \emph{query complexity}. These quantum algorithms were only explicitly related to the classical model of switching networks in~\cite{jeffery_et_al:LIPIcs.STACS.2025.54}, where more detailed techniques were developed for efficiently implementing this type of quantum algorithm, and rigorously analyzing their time complexity. 

Switching networks lend themselves well to quantum algorithms in part because they are a naturally reversible structure, being defined by undirected graphs. Moreover, they give a natural way of defining recursive quantum algorithms: by defining a switching network recursively, and turning the final product into a quantum algorithm, bounded error is introduced only once, when turning the switching network into a quantum algorithm, and there are no factors of $c^d$ or $\log^d$ on the complexity of the algorithm, even when the depth of recursion in the switching network is $d$. 

Because we keep the more space-intensive outer BFS algorithm classical, most of the space needed by the algorithm is classical. For any $S$, we only need $O(\log^2(n))$ qubits of \emph{quantum} space to implement our quantum subroutine for $\textsc{Dist}_L$, the remainder of the $S$ space is classical.

Our algorithm for $\textsc{Dist}_L$ provides a compelling example of how switching networks can be used within quantum algorithms. It exploits the full power of the classical switching network model, while at the same time extending it into a setting where time-space tradeoffs become essential. It is the first time switching networks have been used to prove a quantum time-space tradeoff. 
Beyond its role in our algorithm for \textsc{dstcon}, our algorithm for $\textsc{Dist}_L$ may be of independent interest as a quantum primitive for other space-efficient graph algorithms.

\paragraph{Open Problems:} Our speedup over the best known classical algorithm is quadratic when $S=\log^2(n)$, and gets worse as $S$ increases towards $\sqrt{n}$, at which point we no longer achieve any speedup. However, we do know there is a quadratic speedup at the other extreme, $S=n$. It is thus very natural to hope that we might get a quadratic speedup over the algorithm of~\cite{barnes1998dstconnST} for \emph{all} $S$. One difficulty in this is that our algorithm is not just a quantization of the algorithm of~\cite{barnes1998dstconnST}, but actually does something different. 

As mentioned earlier, for the undirected variant (\textsc{ustcon}), quantum algorithms achieve optimal time and space simultaneously~\cite{belovs2012stcon,apers_et_al:LIPIcs.ESA.2023.10}. This naturally raises the question of whether a similar result is possible in the directed case. 

We have already mentioned the compelling open problem of showing a $O(\log(n))$-space quantum algorithm for \textsc{dstcon}. This would necessarily be a polynomial-time algorithm, and so it would, in some rough sense, achieve the above goal of have optimal space and time simultaneously. It would also show that the quantum analogue of ${\sf L}$ (logspace) is contained in ${\sf NL}$ (nondeterministic logspace). Ref.~\cite{apers_et_al:LIPIcs.CCC.2025.18} made some progress on this question by showing that a promise version of \textsc{dstcon} in which the input has few paths can be solved in $O(\log(n))$ quantum space. We remark that any $o(\log^2(n))$ upper bound on the quantum space complexity of \textsc{dstcon} would be interesting.

\paragraph{Organization:} The remainder of this paper is organized as follows. In \sec{prelim}, we give the necessary preliminaries on graph theory, switching networks, and quantum algorithms. In \sec{BFS}, we present our main result, a quantum algorithm for \textsc{dstcon} with a parameter $L$ that can be used to tune the classical space complexity of the algorithm. The key technical building block of this algorithm is a quantum subroutine for the problem $\textsc{Dist}_L$, which we describe in \sec{short_path_algorithm}.

\section{Preliminaries}\label{sec:prelim}

Here we state preliminaries. In \sec{prelim-alg}, we state some results about efficiently generating quantum superpositions.
In \sec{prelim-graph}, we define both directed and undirected graphs, as well as the notation we will use to talk about them. In \sec{prelim-SN}, we formally define switching networks, and state results about quantum algorithms for evaluating them. Part of instantiating such a quantum algorithm requires generating the states of a basis for a certain graph-theoretic space called the \emph{cut space}, so in \sec{prelim-flows}, we develop some theory that helps with this. 

First, we state a few simple conventions and notational definitions.
Unless otherwise specified, logarithms are with respect to base 2. We let $\bar{0}$ denote the all-zeros string, which throughout this paper will always have length $\log n$.
We use $\tO(f(n))$ to suppress poly$(\log n)$ factors, where $n$ is always the number of vertices in the input digraph. So, for example, $\tO(1)$ is polylog$(n)$.

\subsection{Generating superpositions}\label{sec:prelim-alg}

As part of our algorithms, we will often require a subroutine that makes a superposition over strings, with non-uniform amplitudes. The following lemma gives conditions under which we can do that efficiently.

\begin{lemma}[\cite{grover2002creating}]\label{lem:superposition_preparation}
    Fix $\alpha \in \mathbb{R}^{\{0,1\}^m}$, and suppose there is a subroutine that can compute, given any prefix $p\in \{0,1\}^{\leq m}$, the partial sum of entries of $\alpha$ with prefix $p$, $S(p) = \sum_{s\in\{0,1\}^m:s \text{ has prefix } p}\alpha_s^2$, in time $T$. Then the state $\frac{1}{\sqrt{\sum_s \alpha_s^2}}\sum_s \alpha_s \ket{s}$ can be prepared in time $O(T m)$.
\end{lemma}

\begin{remark}[Larger alphabets]\label{rem:superposition_preparation_012}
The result of \lem{superposition_preparation} can be easily generalized to the case when indices are length-$m$ strings over an alphabet $A$ of constant size $d = \abs{A}$ (e.g., $A = \{0,1,2\}$). 
In this case, $p$ is a prefix over $A$, and the controlled binary rotations from the proof of \cite{grover2002creating} are replaced by $d$-way branching rotations.
Since $d$ is constant, the total complexity remains $O(T m)$.
\end{remark}

\subsection{Graphs}\label{sec:prelim-graph}

A \emph{directed graph} $G=(V(G),E(G))$, or \emph{digraph}, consists of a vertex set 
$V(G)=\{v_1,\dots,v_n\}$
and an edge set $E(G)\subseteq\{(u,v)\in V(G)^2:u\neq v\}$. We say $v$ is \emph{reachable} from $u$ in $G$ if there exists a path $u=u_0,\dots,u_\ell=v$ such that for all $i\in [\ell]$, $(u_{i-1},u_i)\in E(G)$. When $G$ is clear from context, we just write $V$ and $E$.

We are interested in solving the $st$-connectivity -- or reachability -- problem on directed graphs, but we will also make use of undirected graphs, in the context of switching networks, described shortly. 

An undirected graph ${\cal N}$ is just a directed graph, where we don't care about the direction of the edge -- and in fact, we will ultimately want to assign to each edge an arbitrary orientation, for convenience, so we could even take the same definition we used for directed graphs, and just change notions like reachability. However, for our purpose an undirected graph will ultimately define a quantum algorithm on a space spanned by the \emph{edges} of ${\cal N}$. We therefore take a somewhat edge-centric definition, which has the added bonus of accommodating multigraphs. 

\begin{definition}\label{def:graph}
    An \emph{undirected graph} ${\cal N}=(V({\cal N}),E({\cal N}))$ is specified by a pair of finite sets of \emph{vertex labels}, $V$ and \emph{edge labels}, $E$; as well as incidence sets $E_{\cal N}(\u)=E_{\cal N}^{\leftarrow}(\u)\sqcup E_{\cal N}^{\rightarrow}(\u)$ for each $\u\in V$, such that for all $e\in E$, there are unique distinct vertices $\u,\v\in V$ such that $e\in E^{\rightarrow}(\u)\cap E^{\leftarrow}(\v)$. We think of $E^{\rightarrow}(\u)$ as the set of edges \emph{coming out} of $\u$, and $E^{\leftarrow}(\u)$ as the set of edges \emph{going into}~$\u$.
    When ${\cal N}$ is clear from context, we omit it from the notation, as we have just demonstrated.
\end{definition}

This definition gives us the freedom to specify $e\in E$ however we want to, not necessarily by its endpoints, as $e=(\u,\v)$. To recover this usual edge set, we can map $e$ to $(\u,\v)$ such that $e\in E^{\rightarrow}(\u)\cap E^{\leftarrow}(\v)$. When this holds for some edge, we will abuse notation by writing $(\u,\v)\in E$, to mean ``there is an edge from $\u$ to $\v$ in ${\cal N}$.'' Note that while each edge does have an orientation, these are arbitrary and for convenience only. We say that $\u$ is \emph{connected to} $\v$ in ${\cal N}$ (or \emph{reachable from} $\v$) if there is a \emph{$\u\v$-path} $\u=\u_0,\dots,\u_{\ell}=\v$ such that for all $i\in [\ell]$, $(\u_{i-1},\u_i)\in E$ or $(\u_i,\u_{i-1})\in E$.

Finally, we will often build a (undirected) graph ${\cal N}$ out of graphs ${\cal N}_1$ and ${\cal N}_2$ by ``gluing'' (identifying) some of the vertices in ${\cal N}_1$ with some vertices in ${\cal N}_2$. The following definition makes precise what we mean by this. 

\begin{definition}\label{def:gluing}
   Let ${\cal N}_1$ and ${\cal N}_2$ be graphs, with $\u_{1,1},\dots,\u_{1,r}\in V({\cal N}_1)$ and $\u_{2,1},\dots,\u_{2,r}\in V({\cal N}_2)$. The graph ${\cal N}$ obtained from these graphs by \emph{gluing} vertex $\u_{1,i}$ with $\u_{2,i}$ for all $i\in [r]$ is defined by the following vertex and edge (label) sets:
   \begin{align*}
       V({\cal N}) &= V({\cal N}_1)\sqcup
       (V({\cal N}_2)\setminus\{\u_{2,1},\dots,\u_{2,r}\})\\
       \mbox{and }E({\cal N}) &= E({\cal N}_1)\sqcup E({\cal N}_2)
   \end{align*}
   and incidence sets:
   \begin{align*}
       \forall \u\in V({\cal N}_1)\setminus\{\u_{1,1},\dots,\u_{1,r}\},\; E_{\cal N}^{\rightarrow}(\u) &= E_{{\cal N}_1}^{\rightarrow}(\u)
       \mbox{ and }E_{{\cal N}}^{\leftarrow}(\u)=E_{{\cal N}_1}^{\leftarrow}(\u),\\
       \forall \u\in V({\cal N}_2)\setminus\{\u_{2,1},\dots,\u_{2,r}\},\; E_{\cal N}^{\rightarrow}(\u) &= E_{{\cal N}_2}^{\rightarrow}(\u)
       \mbox{ and }E_{{\cal N}}^{\leftarrow}(\u)=E_{{\cal N}_2}^{\leftarrow}(\u),\\
       \mbox{and }\forall i\in [r],\; E^{\rightarrow}_{\cal N}(\u_{1,i}) &= E_{{\cal N}_1}^{\rightarrow}(\u_{1,i})\sqcup E_{{\cal N}_2}^{\rightarrow}(\u_{i,2})
       \mbox{ and } E^{\leftarrow}_{\cal N}(\u_{1,i}) = E_{{\cal N}_1}^{\leftarrow}(\u_{1,i})\sqcup E_{{\cal N}_2}^{\leftarrow}(\u_{i,2}).
   \end{align*}
\end{definition}
We made the arbitrary choice to let the glued vertices inherit their names from ${\cal N}_1$ rather than ${\cal N}_2$, but since in practice, we will mainly work on the edges of the graph, this detail doesn't matter so much. What is more important is that the edge set is simply a disjoint union of the edges sets of ${\cal N}_1$ and ${\cal N}_2$.

\subsection{Switching Networks}\label{sec:prelim-SN}

A switching network on $\{0,1\}^m$ (see e.g.~\cite{potechin2015analyzing}) is an undirected graph $\cal N$ with two distinct \emph{boundary} vertices $\s$ and $\t$, in which each edge is labeled by a literal from $\{x_1,\dots,x_m,\neg x_1,\dots,\neg x_m,1\}$. To simplify things slightly, in our case, we only have edges labeled with positive literals, $\{x_1,\dots,x_m\}$, so we can just use the label set $[m]$ (or more specifically, in our case, the labels will be pairs $i,j\in [n]$, with $x_{i,j}=1$ if and only if $(v_i,v_j)\in E(G)$ for $G$ the input digraph).
For $x\in\{0,1\}^m$, letting ${\cal N}(x)$ denote the subgraph of ${\cal N}$ that includes only those edges whose labels are true under the string $x$, we say that ${\cal N}$ \emph{accepts} $x$ if and only if $\s$ is connected to $\t$ in ${\cal N}(x)$. In~\cite{jeffery_et_al:LIPIcs.STACS.2025.54}, formalizing~\cite{jeffery2017formula}, switching networks were additionally equipped with some associated subspaces, as we describe in the following definition.

\begin{definition}\label{def:switching-network}
A \emph{switching network} ${\cal N}$ on $\{0,1\}^m$ consists of:
\begin{enumerate}
    \item an undirected multigraph ${\cal N}=(V,E)$ with a \emph{source} $\s\in V$ and \emph{sink} $\t\in V\setminus\{\s\}$; 
    \item for each edge $e\in E$, a \emph{query label} $\varphi_e\in [m]$.
\end{enumerate}
For $x\in \{0,1\}^m$, we define ${\cal N}(x)$ by restricted ${\cal N}$ to those edges $e\in E$ such that $x_{\varphi_e}=1$.
We associate the following spaces with ${\cal N}$, and an input $x$:
\begin{enumerate}
    \item  
    $\Xi_{\s}^{\cal A} = \mathrm{span}\{\ket{\s}+\ket{\leftarrow,\s}\}$ and 
    $\Xi_{\t}^{\cal A} = \mathrm{span}\{\ket{\rightarrow,\t}+\ket{\t}\}$
    \item for all $\v\in V\setminus\{\s,\t\}$, ${\cal V}_{\v}=\mathrm{span}\{\ket{\psi_\star(\v)}\colonequals \sum_{e\in E^{\rightarrow}(\u)}\ket{\rightarrow,e}+\sum_{e\in E^{\leftarrow}(\u)}\ket{\leftarrow,e}\}$
    \item ${\cal V}_{\s}=\mathrm{span}\{\ket{\psi_\star(\s)}\colonequals \sum_{e\in E^{\rightarrow}(\s)}\ket{\rightarrow,e}+\sum_{e\in E^{\leftarrow}(\s)}\ket{\leftarrow,e}+\ket{\leftarrow,\s}\}$
    \item ${\cal V}_{\t}=\mathrm{span}\{\ket{\psi_\star(\t)}\colonequals \sum_{e\in E^{\rightarrow}(\t)}\ket{\rightarrow,e}+\sum_{e\in E^{\leftarrow}(\t)}\ket{\leftarrow,e}+\ket{\rightarrow,\t}\}$
    \item for all $e\in E$, $\Xi_e=\mathrm{span}\{\ket{\rightarrow,e},\ket{\leftarrow,e}\}$, $\Xi_e^{\cal A}(x)=\mathrm{span}\{\ket{\rightarrow,e}+(-1)^{x_{\varphi_e}}\ket{\leftarrow,e}\}$, $\Xi_e^{\cal B}=\mathrm{span}\{\ket{\rightarrow,e}+\ket{\leftarrow,e}\}$.
\end{enumerate}
Then we let $H_{\cal N}=\bigoplus_{e\in E}\Xi_e \oplus \mathrm{span}\{\ket{\s},\ket{\t},\ket{\leftarrow,\s},\ket{\rightarrow,\t}\}$ and define the following two important subspaces of $H_{\cal N}$:
\begin{equation}\label{eq:calA-calB}
    {\cal A}(x)=
    \Xi_{\s}^{\cal A}\oplus\Xi_{\t}^{\cal A}\oplus \bigoplus_{e\in E}\Xi_e^{\cal A}(x)
    \quad\mbox{and}\quad
    {\cal B}=\bigoplus_{\u\in V(G)}{\cal V}_{\u} + \bigoplus_{e\in E}\Xi_e^{\cal B}\oplus\mathrm{span}\{\ket{\leftarrow,\s}+\ket{\rightarrow,\t}\}.
\end{equation}
\end{definition}

A single switching network can potentially compute different functions for different possible values of ${\sf s}$ or ${\sf t}$. Thus, we will later find it convenient to equip a \emph{set} of sinks\footnote{We could also allow multiple sources instead of just a single ${\sf s}$, but we don't need that for our constructions.}, which, together with the single source ${\sf s}$, will form a \emph{boundary} for ${\cal N}$. 

\begin{remark}
    It is also possible to assign weights to the edges of ${\cal N}$ in the definition of a switching network, which impacts the complexity of the quantum algorithm for evaluating it, but in this work we will assume all edges have weight 1. 
\end{remark}

A switching network, along with the spaces defined in \defin{switching-network}, is a special case of a \emph{subspace graph}, defined in \cite{jeffery_et_al:LIPIcs.STACS.2025.54}. As there, we can define \emph{working bases} for a switching network as a pair of bases $\Psi^{\cal A}$ and $\Psi^{\cal B}$ for ${\cal A}$ and ${\cal B}$ respectively. For switching networks, the natural working basis for ${\cal A}$ is simply
$$\Psi_{\cal A}(x)=\left\{\ket{\rightarrow,e}+(-1)^{x_{\varphi_e}}\ket{\leftarrow,e}:e\in E(\cal N) \right\}\cup\{\ket{\s}+\ket{\leftarrow,\s},\ket{\rightarrow,\t}+\ket{\t}\}.$$
It's not difficult to see that this basis can be generated using a query to $x$, and $O(1)$ basic operations, by which we mean the following.
\begin{definition}[Basis Generation]
    We say an orthonormal basis $\Psi=\{\ket{b_\ell}\}_{\ell\in L}\subset H_{\cal N}$ can be generated in time $T$ if:
    \begin{enumerate}
        \item The reflection around the subspace span$\{\ket{\ell}:\ell\in L\}$ of $H_{\cal N}$ can be implemented in time $T$.
        \item There is a map that acts as $\ket{\ell}\mapsto \ket{b_\ell}$ for all $\ell\in L$ that can be implemented in time $T$.
    \end{enumerate}
\end{definition}
For the space ${\cal B}$, we construct a basis for the orthogonal complement and invoke the following lemma from \cite[Corollary 2.11]{jeffery_et_al:LIPIcs.STACS.2025.54}, which guarantees that this suffices.
\begin{lemma}\label{lem:dual-basis-gen}
If $\Psi$ is a basis that can be generated in time $T$, then there is a basis $\Psi'$ for $\mathrm{span}\{\Psi\}^\bot$ that can be generated in time $T$.
\end{lemma}
The following is a special case of \cite[Theorem 3.13]{jeffery_et_al:LIPIcs.STACS.2025.54}, which is proven by analyzing a phase estimation algorithm on the product of reflections around ${\cal A}$, and ${\cal B}$. It is also similar to \cite[Theorem 13]{jeffery2017formula}, but with a different implementation of reflection around ${\cal B}$ by generating a basis directly, rather than via a quantum walk, which is potentially more expensive.

\begin{theorem}\label{thm:SN_algorithm}
    For $f:\{0,1\}^m\rightarrow\{0,1\}$, let ${\cal N}$ be a switching network that accepts $x\in\{0,1\}^m$ if and only if $f(x)=1$. Suppose that for all $x$ accepted by ${\cal N}$, there is an $\s\t$-path in ${\cal N}(x)$ of length at most $W_+$.
    Suppose the space ${\cal B}^\perp$ of $\cal N$ has a basis $\Psi_{\cal B}^\perp$ that can be generated in time $T_B$.
    Then there is a quantum algorithm that decides $f$ with bounded error in time
    $O(T_B\sqrt{W_+|E({\cal N}|})$
    and space
    $O(\log |E({\cal N}|).$
\end{theorem}
We can get a stronger version of this theorem by replacing path length with effective resistance, and cut size with capacitance (see \cite{jarret2018connectivity}), but the theorem as stated suffices for our purposes. 

We will use the following statement to find the orthogonal complement of $\cal B$. This is a slight generalization of \cite[Lemma 3.15]{jeffery_et_al:LIPIcs.STACS.2025.54}.
\begin{lemma}\label{lem:B_space_basis}
    Fix a switching network, and define, for any $u\in V({\cal N})$, 
    $$\ket{\psi^-_\star(\u)}:=\sum_{e\in E^\rightarrow(\u)}\frac{1}{2}(\ket{\rightarrow,e}-\ket{\leftarrow,e})+\sum_{e\in E^{\leftarrow}(\u)}\frac{1}{2}(\ket{\leftarrow,e}-\ket{\rightarrow,e}).$$
    Define the \emph{cut space} of $\cal N$ as
    $${\cal B}^{-}=\mathrm{span}\{\ket{\psi_\star^-(\u)}:\u\in V({\cal N})\setminus \{\s,\t\}\}\cup\{\ket{\leftarrow,\s}+\ket{\psi_\star^-(\s)},\ket{\rightarrow,\t}+\ket{\psi_\star^-(\t)}\}.$$
    Then if $\Psi_{\cal B}^-$ is an orthonormal basis for ${\cal B}^-$, the following is an orthonormal basis for ${\cal B}$:
    $$\Psi_{\cal B}=\Psi_{\cal B}^-\cup\left\{\ket{b_e}:=\frac{1}{\sqrt{2}}(\ket{\rightarrow,e}+\ket{\leftarrow,e}):e\in E({\cal N})\right\}.$$
\end{lemma}

Since the map $\ket{e}\ket{0}\mapsto \ket{b_e}$ can be implemented with a single Hadamard gate, it is sufficient to be able to generate an orthonormal basis for ${\cal B}^-$. In the final section of the preliminaries, we give some results that further simplify this task. 

\subsection{Flows, Circulations, and the Cut Space}\label{sec:prelim-flows}

In this section, we will study the structure of \emph{cut spaces}, of undirected graphs, with respect to the boundary $\{\s,\t\}$:
    $${\cal B}^{-}=\mathrm{span}\{\ket{\psi_\star^-(\u)}:\u\in V({\cal N})\setminus \{\s,\t\}\}\cup\{\ket{\leftarrow,\s}+\ket{\psi_\star^-(\s)},\ket{\rightarrow,\t}+\ket{\psi_\star^-(\t)}\}.$$
This is called a cut space, because it is intuitively the span of \emph{cuts}, or sets of edges whose removal leaves the graph disconnected.
We first show the following simple fact.
\begin{lemma}\label{lem:dim-B-minus}
    $\dim{\cal B}^- = |V({\cal N})|$.
\end{lemma}
\begin{proof}
 We show this by proving that the states in the definition of ${\cal B}^-$ are independent. Let 
    $$\ket{\psi_\u}=\left\{\begin{array}{ll}
        \ket{\psi_\star^-(\u)} & \mbox{if }\u\in V \setminus \{\s,\t\}\\
        \ket{\leftarrow,\s}+\ket{\psi_\star^-(\s)} & \mbox{if }\u=\s\\
        \ket{\rightarrow,\t}+\ket{\psi_\star^-(\t)} & \mbox{if }\u=\t\\
    \end{array}\right.$$
    and for simplicity, write $\ket{e}=\ket{\rightarrow,e}-\ket{\leftarrow,e}$. 
    Suppose towards a contradiction that for some $\u\in V$,
    $$\ket{\psi_{\u}}=\sum_{\v\in V\setminus\{\u\}}\alpha_{\v}\ket{\psi_{\v}}.$$
    We will show by induction on distance $d\geq 1$ from $\u$, that for all $\u'$ holds $\abs{\alpha_{\u'}} = 1$. For the base case, if $\u'$ and $\u$ share an edge $e$, then
    $$1=\abs{\braket{e}{\psi_\u}}=\abs{\sum_{\v\in V\setminus\{\u\}}\alpha_{\v}\braket{e}{\psi_\v}}=\abs{\alpha_{\u'}},$$
    since $\abs{\braket{e}{\psi_\v}}$ is 1 if $e$ is incident to $\v$ and 0 otherwise. 

    For the induction step, if $\u'$ is at distance $d$ from $\u$, it has a neighbour $\u''$ at distance $d-1$ from $\u$, to which the induction hypothesis applies, so letting $e$ be the edge incident to $\u'$ and $\u''$, we have:
    $$0=\abs{\sum_{\v\in V\setminus\{\u\}}\alpha_\v\braket{e}{\psi_\v}}=\abs{\alpha_{\u'}+\alpha_{\u''}}$$
    so $\abs{\alpha_{\u'}}=1$, since $\abs{\alpha_{\u''}}=1$ by the induction hypothesis. 
    
    The contradiction arises because there is at least one boundary vertex that is not $\u$ -- without loss of generality, suppose $\s\neq\u$. Then since $\ket{\leftarrow,\s}$ only appears in $\ket{\psi_\s}$, we have:
    $$0=\abs{\braket{\leftarrow,\s}{\psi_\u}}=\abs{\sum_{\v\in V\setminus\{\u\}}\alpha_{\v}\braket{\leftarrow,\s}{\psi_{\v}}}=\abs{\alpha_{\s}}=1.$$
    This is clearly a contradiction.
\end{proof}

\begin{definition}\label{def:flow}
    A \emph{flow} on ${\cal N}$ is a real-valued function $\theta$ on $E({\cal N})$. For $\u \in V({\cal N})$, define 
    $$\theta(\u)\colonequals\sum_{e\in E^{\rightarrow}(\u)}\theta(e)-\sum_{e\in E^\leftarrow(\u)}\theta(e).$$
    We call $\theta$ a \emph{circulation} on ${\cal N}$ if for all $\u\in V({\cal N})$, $\theta(\u)=0$. 
We say $\theta$ has \emph{boundary} $B$ if for all $\u\in V({\cal N})\setminus B$, $\theta(\u)=0$.
We call $\theta$ an \emph{$\s\t$-flow} if it has boundary $\{\s,\t\}$ (so in particular, every circulation is an $\s\t$-flow).
    We call $\theta$ a \emph{unit} $\s\t$-flow if it is an $\s\t$-flow, and $ \theta(\s)=-\theta(\t)=1$. We call $\theta$ an \emph{optimal} unit $\s\t$-flow if it minimizes the expression $\sum_{e\in E({\cal N})}{\theta(e)^2}$. 
\end{definition}    

We can naturally view any function $\theta$ on $E({\cal N})$ as a vector on
$$\mathrm{span}\{\ket{e}:=\ket{\rightarrow,e}-\ket{\leftarrow,e}:e\in E({\cal N})\},$$
which we denote
\begin{equation}\label{eq:cropped}
    \ket{\bar\theta}:=\sum_{e\in E({\cal N})}{\theta(e)}(\ket{\rightarrow,e}-\ket{\leftarrow,e})=\sum_{e\in E({\cal N})}{\theta(e)}\ket{e}.
\end{equation}
For an edge $e$ oriented from $\u$ to $\v$, if $\theta(e)$ is a positive real number, we interpret $\theta$ as sending $\theta(e)$ flow from $\u$ to $\v$ along the edge $e$. If $\theta(e)$ is negative, we interpret $\theta$ as sending $|\theta(e)|$ flow from $\v$ to $\u$. In this way, we can interpret $-\ket{e}$ as the edge $e$ oriented in the opposite direction -- that is, vector negation changes the orientation of edges. 

When $\theta$ is an $\s\t$-flow, we will sometimes want to additionally include flow entering and exiting the boundary vertices, via ``boundary edges'', so that the total flow is conserved on boundary vertices as well, and this is represented by the vector
\begin{equation}\label{eq:flow_state}
    \ket{\theta} = -\theta(\s)\ket{\leftarrow,\s}+\ket{\bar\theta}-\theta(\t)\ket{\rightarrow,\t}.
\end{equation}
Whereas states of the form $\ket{\bar\theta}$ enable seamless combination of flows in graphs obtained from gluing other graphs together (see \defin{gluing}), states of the form $\ket{\theta}$ capture the boundary of the final graph we use in our switching network.

It is not difficult to see that the set of flows on a particular (possibly empty) boundary is closed under linear combinations, so we can define the following vectors spaces.
\begin{definition}\label{def:flow-spaces}
    We call
    \begin{align*}
\mathcal{F}({\cal N}) = \Big\{& \ket{\theta} :\forall \u\in V({\cal N}) \setminus \{\s,\t\},\; \theta(\u)=0 \Big\}\end{align*}
the space of $\s\t$-flows;
\[\mathcal{C}({\cal N}) = \Big\{ \ket{\theta}:
\forall \u\in V({\cal N}),\; \theta(\u)=0 \Big\}\]
the space of circulations;
\[\mathcal{F}^\textsc{OPT}({\cal N}) = \mathrm{span}\left\{ \ket{\theta} \in \mathcal{F}({\cal N}): \ket{\theta} \text{ is an optimal unit $\s\t$-flow} \right\}\]
the space of optimal $\s\t$-flows.
\end{definition}

In the following lemmas, we establish fundamental properties of the space of $\s\t$-flows together with its subspaces. These results will serve as the foundation for our later analysis. We start with stating some properties of optimal flows and their orthogonality to circulations, see \cite[Lemma~7.2.2]{cornelissen2023thesis} for the proof. 
\begin{lemma}\label{lem:flow_properties}
    The optimal unit $\s\t$-flow $\theta$ is unique and its corresponding vectors, $\ket{\theta}$ and $\ket{\bar\theta}$ (see \eq{cropped} and \eq{flow_state}) are orthogonal to the space of circulations $\mathcal{C}({\cal N})$. Consequently, the space of optimal $\s\t$-flows $\mathcal{F}^\textsc{OPT}({\cal N})$ is one-dimensional and orthogonal to the space of circulations $\mathcal{C}({\cal N})$.
\end{lemma}
Next, we show that the space of $\s\t$-flows admits a direct sum decomposition into optimal flows and circulations.
\begin{lemma}\label{lem:opt_flows_oplus_circulations}
    $\mathcal{F}({\cal N}) = \mathcal{F}^{\textsc{OPT}}({\cal N}) \oplus \mathcal{C}({\cal N})$
\end{lemma}
\begin{proof}
    By \lem{flow_properties}, $\mathcal{F}^{\textsc{OPT}}(\mathcal{N}) \perp \mathcal{C}(\mathcal{N})$. It is immediate that $\mathcal{F}({\cal N}) \supseteq \mathcal{F}^{\textsc{OPT}}({\cal N}) \oplus \mathcal{C}({\cal N})$. Thus, it remains to prove the reverse inclusion $\mathcal{F}({\cal N}) \subseteq \mathcal{F}^{\textsc{OPT}}({\cal N}) \oplus \mathcal{C}({\cal N})$.
    Let $\ket{\theta} \in \mathcal{F}({\cal N})$. If $\theta$ is a circulation, then 
    we're done, so assume not. Without loss of generality, assume $\theta$ is scaled so that $\theta(\s)=-\theta(\t)=1$. Then, by \lem{flow_properties}, there exists a unique optimal unit $\s\t$-flow $\ket{\theta'} \in \mathcal{F}^{\textsc{OPT}}$. It is easy to check that $\ket{\theta} - \ket{\theta'}$ is a circulation, which establishes $\ket{\theta}\in \mathcal{F}^{\textsc{OPT}}({\cal N}) \oplus \mathcal{C}({\cal N})$, completing the proof.
\end{proof}

Next, we establish the relationship between the space of $\s\t$-flows of a switching network and its associated space $\mathcal{B}$.
\begin{lemma}\label{lem:B_perp}
Let ${\cal N}$ be a switching network, with ${\cal B}$ as in \eq{calA-calB}. Then $\mathcal{B}^\perp = \mathcal{F}({\cal N}) \oplus \mathrm{span} \{ \ket{\s},\ket{\t}\}$. 
\end{lemma}
\begin{proof}
First, we observe that $\mathcal{F}({\cal N})$ is orthogonal to $ \mathrm{span} \{ \ket{\s},\ket{\t}\}$, because, by definition, flow states have no overlap with $\ket{\s}$ and $\ket{\t}$.
By \lem{B_space_basis} we can write:
    \begin{align*}
        \cal B^\perp 
        &=\left({\cal B}^-\oplus\mathrm{span}\left\{\frac{1}{\sqrt{2}}(\ket{\rightarrow,e}+\ket{\leftarrow,e}):e\in E(G)\right\}\right)^\perp \\
        &= \left({\cal B}^-\right)^\perp \cap \left(\mathrm{span}\left\{\frac{1}{\sqrt{2}}(\ket{\rightarrow,e}+\ket{\leftarrow,e}):e\in E(G)\right\}\right)^\perp.
    \end{align*}
Note that 
\begin{align*}
  &  \left(\mathrm{span}\left\{\frac{1}{\sqrt{2}}(\ket{\rightarrow,e}+\ket{\leftarrow,e}):e\in E(G)\right\}\right)^\perp\\
  ={}& \mathrm{span}\left\{\frac{1}{\sqrt{2}}(\ket{\rightarrow,e}-\ket{\leftarrow,e}):e\in E(G)\right\}\oplus \mathrm{span}\{\ket{\s},\ket{\t},\ket{\leftarrow,\s},\ket{\rightarrow,\t}\}.
\end{align*}
Since ${\cal B}^-\subseteq \left(\mathrm{span}\left\{\frac{1}{\sqrt{2}}(\ket{\rightarrow,e}+\ket{\leftarrow,e}):e\in E(G)\right\}\right)^\perp$, we can take the orthogonal complement of ${\cal B}^-$ in $\left(\mathrm{span}\left\{\frac{1}{\sqrt{2}}(\ket{\rightarrow,e}+\ket{\leftarrow,e}):e\in E(G)\right\}\right)^\perp$. That is, we can assume that all edge spaces are one-dimensional and are spanned by $\ket{e} = \ket{\rightarrow,e}-\ket{\leftarrow,e}$, for $e\in E$.  We can take an arbitrary state $\ket{\psi} \in \left(\mathrm{span}\left\{\frac{1}{\sqrt{2}}(\ket{\rightarrow,e}+\ket{\leftarrow,e}):e\in E(G)\right\}\right)^\perp$ and write it as $\ket{\psi} = \ket{\hat\theta} + \ket{b}$, where
\begin{align*}
    \ket{\hat\theta} &= -\theta_{\s}\ket{\leftarrow,\s}+\sum_{e\in E}{\theta(e)}\ket{e}-\theta_{\t}\ket{\rightarrow,\t},
\end{align*}
for some function $\theta$ on $E$ and values $\theta_{\s},\theta_\t$, and $\ket{b}\in\mathrm{span}\{\ket{\s},\ket{\t}\}$.
By definition of $\ket{\psi_\star^-(\u)}$ (see \lem{B_space_basis}), it is orthogonal to $\ket{b}$ for any $\u \in V$, and so by inspection, $\ket{b}$ is orthogonal to ${\cal B}^-$.

Next, we compute the inner product of $\ket{\hat\theta}$ with the states in ${\cal B}^-$ to determine precisely when $\ket{\psi}$ is orthogonal to all of them.

First, for \emph{any} $\u\in V$,
\begin{equation}\label{eq:vanish0}
\begin{split}
    \braket{\hat\theta}{\psi_\star^-(\u)} &= \sum_{e\in E^\rightarrow (\u)} \frac{1}{2}\theta (e) \braket{e}{e} - \sum_{e\in E^\leftarrow (\u)} \frac{1}{2} \theta (e) \braket{e}{e}\\
    &= \sum_{e\in E^\rightarrow (\u)} \theta (e) - \sum_{e\in E^\leftarrow (\u)} \theta (e) = \theta(u).
\end{split}
\end{equation}
This inner product vanishes for all $\u\in V\setminus \{\s,\t\}$ if and only if $\theta$ is an $\s\t$-flow.
Next, we compute:
\begin{align}\label{eq:vanish1}
    \bra{\hat\theta}(\ket{\leftarrow,\s}+\ket{\psi_\star^-(\s)}) &= -\theta_{\s}+\braket{\hat\theta}{\psi_\star^-(\s)}=-\theta_{\s}+\theta(\s), 
\end{align}
by \eq{vanish0}.
This inner product vanishes if and only if $\theta_{\s}=\theta (\s)$.
Finally, we compute: \begin{align}\label{eq:vanish2}
    \bra{\theta}(\ket{\rightarrow,\t}+\ket{\psi_\star^-(\t)}) &= -\theta_{\t}+\theta(\t),
\end{align}
again by \eq{vanish0}. This inner product vanishes if and only if $\theta_\t=\theta(\t)$.

Combining \eq{vanish0}, \eq{vanish1} and \eq{vanish2}, we can see that $\ket{\psi}\in {\cal B}^-$ if and only if $\ket{\hat\theta}=\ket{\theta}$ (as defined in \eq{flow_state}) for some $\s\t$-flow $\theta$, which is if and only if $\ket{\hat\theta}\in {\cal F}({\cal N})$.
\end{proof}

Finally, we compute the dimensions of the space of flows and circulations of ${\cal N}$ and its subspaces. 
\begin{lemma}\label{lem:flow_dimensions}
Let $\mathcal{F}({\cal N})$ be the space of boundary flows of ${\cal N}$; $\mathcal{C}({\cal N})$ the space of circulations of ${\cal N}$; and $\mathcal{F}^\textsc{OPT}({\cal N})$ the span of optimal boundary flows of ${\cal N}$ (see \defin{flow-spaces}). Then
\begin{align*}
    \dim \mathcal{F}({\cal N}) &= \abs{E({\cal N})} +2 -\abs{V({\cal N})} \\
    \dim \mathcal{F}^\textsc{OPT}({\cal N}) &= 1\\
    \dim \mathcal{C}({\cal N}) &= \abs{E({\cal N})} -\abs{V({\cal N})} + 1.
\end{align*}
\end{lemma}
\begin{proof}
    The dimension of $\mathcal{F}({\cal N})$ can be computed using \lem{B_perp}:  
\[
\dim \mathcal{F}({\cal N}) + \dim\mathrm{span} \{ \ket{\s},\ket{\t}\} = \dim H_{{\cal N}} - \dim \mathcal{B}.
\]
The dimension of $H_{{\cal N}}$ is equal to $2\abs{E({\cal N})} + 4$ by \defin{switching-network}. 
By \lem{B_space_basis} and \lem{dim-B-minus}, $\dim \mathcal{B} = \dim{\cal B}^-+\abs{E({\cal N})} = \abs{V({\cal N})} + \abs{E({\cal N})}$. Combining these observations,
we obtain 
\begin{align*}
    \dim \mathcal{F}({\cal N}) &= \dim H_{{\cal N}} - \dim \mathcal{B} - 2
    = \abs{E({\cal N})} + 2 - \abs{V({\cal N})}.
\end{align*}
Next, we note that it follows from \lem{flow_properties} that $\dim \mathcal{F}^{\textsc{OPT}}({\cal N}) = 1$.
Finally, the dimension of $\mathcal{C}({\cal N})$ is straightforward to determine, as it follows from the direct sum decomposition $\mathcal{F}({\cal N}) = \mathcal{F}^{\textsc{OPT}}({\cal N}) \oplus \mathcal{C}({\cal N})$ from \lem{opt_flows_oplus_circulations}.
\end{proof}

\section{Quantum Algorithm for DSTCON}\label{sec:BFS}

In this section, we prove our main result, by describing a quantum algorithm for $\textsc{dstcon}$ that works for any space bound $S\geq \log^2(n)$. 
We begin by stating the main technical result of this paper, which we prove in \sec{short_path_algorithm}. Specifically, we present a quantum algorithm for deciding directed $st$-connectivity under an additional constraint on the path length (i.e.~solving \textsc{Path}$_L$). While this subroutine is quantum, the algorithm we decide in the remainder of this section is otherwise classical. 
\begin{theorem}\label{thm:quantum_short_path_complexity}
Let $G=(V,E)$ be a directed graph such that $V = \{v_1,\dots,v_n\}$. Assume that $G$ can be accessed via a quantum oracle $\mathcal{O}_G$ that can be implemented in time $O(1)$, where for any $i,j\in [n]$, $b\in\{0,1\}$ 
$$\mathcal{O}_G: \ket{i}\ket{j}\ket{b}\mapsto \begin{cases}
   \ket{i}\ket{j}\ket{b\oplus 1} & \mbox{if } (v_i,v_j)\in E\\
    \ket{i}\ket{j}\ket{b} & \mbox{otherwise.}
\end{cases}$$
Let $L\leq n$ be a power of 2. Then there is a bounded-error quantum algorithm, ${\tt D}_L(G,u,v)$, that decides for any $u,v \in V$ whether there is a directed  path from $u$ to $v$ in $G$ of length at most $L$, in time
$$\tO \left( \left( L^{\log 3} (2n+1)^{\log L} n \right)^{1/2} \right)$$
and space
$O \left( \log (L) \log (n) \right).$
\end{theorem}

In \sec{short_path_algorithm}, we prove the statement for $L$ a power of 2, using a recursive structure, but a simple corollary extends this result to any $L$.

\begin{corollary}\label{cor:Dist}
    Let $G$ be as in \thm{quantum_short_path_complexity}, and let $L$ be \emph{any} positive integer. Then there is a bounded-error quantum algorithm, ${\tt Dist}_L(G,u,v)$, that decides for any $u,v \in V$ whether there is a directed  path from $u$ to $v$ in $G$ of length at most $L$, in time
$$\tO \left( \left( L^{\log 3} (2n+1)^{\log L} n \right)^{1/2} \right)$$
and space
$O \left( \log (L) \log (n) \right).$
\end{corollary}
\begin{proof}
    To prove the statement, we exhibit a quantum algorithm, \algo{InnCheck} that uses one call to the subroutine ${\tt D}$ from \thm{quantum_short_path_complexity}, on a graph $G'$ that is constructed from $G$ by adding a directed path from some new vertex $s_1$ into $u$ of length such that any $uv$-path of length at most $L$ corresponds to a $s_1v$-path of length at most $2^{\lceil\log L\rceil}$ (see \fig{Gprime}). Since $G'$ can be queried using at most one query to $G$, the result follows.  
\end{proof}

\begin{algorithm}[h]
\caption{{\tt Dist}$_L(G,u,v)$}\label{algo:InnCheck}
Parameter: a positive integer $L\leq n$\\
Input: a directed graph $G=(V,E)$ and a pair of vertices $u,v \in V$\\
Output: 0 or 1 indicating there is a path of length at most $L$ between $u$ and $v$ in $G$
\vskip3pt
\hrule
\vskip2pt
\hrule
\vskip5pt
\begin{enumerate}
    \item Let $\ell=\lceil\log L\rceil$ and  $a=2^\ell-L$
    \item Let $G'$ be the graph $G$ with $a$ new vertices $s_1,\dots,s_a$ and $a$ new edges $(s_1,s_2),\dots,(s_{a-1},s_a),(s_a,u)$. Then $G'$ is just $G$ with a directed path of length $a$ coming into $u$, and can easily be queried using queries to $G$.
    \item Return ${\tt D}_{2^\ell}(G,s_1,v)$.
\end{enumerate}
\end{algorithm}

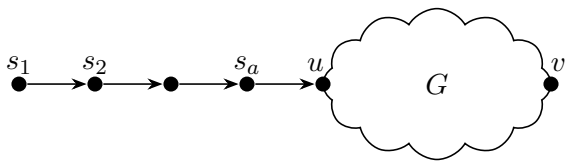
\begin{figure}
    \centering
\begin{tikzpicture}[>={Stealth},       
                    auto,
                    node distance=1cm,  
                    semithick]          
  \node (v1) [circle, fill, inner sep=2pt] {};
  \node (v2) [circle, fill, inner sep=2pt, right of=v1] {};
  \node (v3) [circle, fill, inner sep=2pt, right of=v2] {};
  \node (v4) [circle, fill, inner sep=2pt, right of=v3] {};
  \node (u) [circle, fill, inner sep=2pt, right of=v4] {};
  \node (v) [circle, fill, inner sep=2pt, right of=u, xshift=2cm] {};
  
  \path[->] (v1) edge (v2)
            (v2) edge (v3)
            (v3) edge (v4)
            (v4) edge (u);

  \node[draw,
        cloud,
        cloud puffs=12,
        cloud puff arc=120,
        aspect=2,
        minimum width=3cm,
        minimum height=2cm,
        inner sep=0pt,
        name=cloud] at (5.5,0) {$G$};

    \node at (0,.25) {$s_1$};
    \node at (1,.25) {$s_2$};
    \node at (3,.25) {$s_a$};
    \node at (3.9,.25) {$u$};
    \node at (7.1,.25) {$v$};
\end{tikzpicture}
    \caption{The graph $G'$ constructed from $G$. It is clear that there is a $uv$-path of length at most $L$ in $G$ if and only if there is a $s_1v$-path of length at most $L+a$ in $G'$.}
    \label{fig:Gprime}
\end{figure}

\subsection{BFS algorithm that calls the quantum short path subroutine}

To get a time-space tradeoff for \textsc{dstcon}, and prove our main result, we describe a classical BFS-based algorithm, \algo{outer}, from~\cite{barnes1998dstconnST}, that makes calls to a subroutine for the problem \textsc{Dist}$_L(G,u,v)$, of deciding whether there is a path of length at most $L$ from $u$ to $v$ in a directed graph $G$. Our algorithm for \textsc{dstcon} is obtained by instantiating that subroutine with the quantum algorithm ${\tt Dist}_L$ from \cor{Dist}.

\begin{algorithm}[h]
\caption{{\tt DSTCON}$_{L}(G,s,t)$ \cite{barnes1998dstconnST}}\label{algo:outer}
Parameter: a positive integer $L\leq n$\\
Input: a directed graph $G=(V,E)$ and a pair of vertices $s,t \in V$\\
Output: \textsc{Connected} if there is a directed path from $s$ to $t$ in $G$, \textsc{Not Connected} otherwise
\vskip3pt
\hrule
\vskip2pt
\hrule
\vskip5pt
\begin{algorithmic}[1]
    \For{$j = 0 ,\ldots, L-1$}
        \State $S \gets \{ s \}$
        \For{all vertices $v \in V$}
            \If{ ${\tt Dist}_j (s,v) = 1 \land {\tt Dist}_{j - 1} (s,v) = 0$}
                \If{$|S| > n / L$}
                     try next $j$
                \Else
                    \: add $v$ to $S$ \EndIf
            \EndIf
        \EndFor
        \For{$i = 1 ,\ldots, \lfloor n/L \rfloor$}
            \State $S' = \varnothing$
            \For{all vertices $v \in V$}
                \If{ $ \exists u \in S:{\tt Dist}_L (u,v) = 1 \land \forall u \in S:{\tt Dist}_{L - 1} (u,v) = 0$}
                    \If{$\abs{S} + \abs{S'} > n / L$}
                         try next $j$
                    \Else
                        \: add $v$ to $S'$
                    \EndIf
                \EndIf
            \EndFor
            \State $S = S \cup S'$
        \EndFor
        \If{$t$ within distance $L$ of a vertex in $S$}
             \Return (\textsc{Connected})
        \Else 
            \: \Return (\textsc{Not Connected})
        \EndIf
    \EndFor
\end{algorithmic}
\end{algorithm}

\noindent Ref.~\cite{barnes1998dstconnST} show that this algorithm correctly decides \textsc{dstcon} whenever ${\tt Dist}_L$ decides $\textsc{Path}_L$.
From~\cite{barnes1998dstconnST}, or by inspecting \algo{outer}, we get the following.
\begin{lemma}
Let ${\tt Dist}_L(G,u,v)$ be a bounded-error algorithm for $\textsc{Dist}_L$, with time complexity  $D_T(n,L)$, and the space complexity is $D_S(n,L)$. Then the time complexity of \algo{outer} is
    $$\tO\left(\frac{n^3}{L} D_T(n,L)\right)$$
    and its space complexity is
    $$O\left(\frac{n \log L}{L}+D_S(n,L)\right).$$ 
    If ${\tt Dist}_L$ is implemented by a quantum algorithm, then the \emph{quantum} space complexity is at most $O(D_S(n,L))$, and any remaining space is classical. 
\end{lemma}
Substituting the result of \cor{Dist} yields the following theorem, which combines our quantum subroutine with the classical BFS algorithm. Although the algorithm in \cor{Dist} has bounded error, its success probability can be boosted high enough that the outer algorithm will not notice, using majority voting, at the cost of an overhead of $\log \tO(n^3/L)$, which is hidden in the $\tO$ of the final complexity.  
\begin{theorem}\label{thm:quantum_dstcon}
Let $G=(V,E)$ be a directed graph such that $V = \{v_1,\dots,v_n\}$, and $s,t \in V$. Assume that $G$ can be accessed via a quantum oracle $\mathcal{O}_G$ that can be implemented in time $O(1)$, where for any $i,j\in [n]$, $b\in\{0,1\}$, 
$$\mathcal{O}_G: \ket{i}\ket{j}\ket{b}\mapsto \begin{cases}
   \ket{i}\ket{j}\ket{b\oplus 1} & \mbox{if } (v_i,v_j)\in E\\
    \ket{i}\ket{j}\ket{b} & \mbox{otherwise.}
\end{cases}$$
Then there is a quantum algorithm that decides whether there is a directed  path from $s$ to $t$ in $G$ with bounded error in time
$$\tO \left( n^{3.5} L^{.5\log(3)-1} (2n+1)^{0.5\log L} \right)$$
and total space
$$\tO \left( \left( \frac{n}{L} + \log L \right)\log n \right)$$
of which 
$O(\log(L)\log(n))$ is quantum space.
\end{theorem}

\subsection{Complexity comparison: classical vs. quantum}

\begin{theorem}[\cite{barnes1998dstconnST}]\label{thm:classical_tradeoff}
    Let $G=(V,E)$ be a directed graph such that ${V} = \{v_1,\dots,v_n\}$, and $s,t \in V$. Assume that $G$ can be accessed via a classical oracle ${\cal O}_G$ that can be implemented in time $O(1)$, where for any $i,j\in [n]$, 
    $${\cal O}_G(i,j)=\begin{cases}
   1 & \mbox{if } (v_i,v_j)\in E\\
    0 & \mbox{otherwise.}
\end{cases}$$
    Then for any $S\geq \log^2(n)$, there is a classical algorithm that decides whether there is a directed path from $s$ to $t$ in $G$ using space $O(S)$ and time $$T \leq 2^{\log^2 \frac{n}{S} +O(\log n\log\log n)}.$$
\end{theorem}
In \cite{barnes1998dstconnST}, they state a bound of $T\leq 2^{O(\log^2(n/S))}$, but using their choice of parameters, and a slightly more precise analysis of their algorithm, we can compute the more fine-grained upper bound we have stated above. We improve on their result in our main theorem, which is the following. 
\begin{theorem}\label{thm:quantum_tradeoff}
Let $G=(V,E)$ be a directed graph such that $V = \{v_1,\dots,v_n\}$, and $s,t \in V$. Assume that $G$ can be accessed via a quantum oracle $\mathcal{O}_G$ that can be implemented in time $O(1)$, where for any $i,j\in [n]$, $b\in\{0,1\}$, 
$$\mathcal{O}_G: \ket{i}\ket{j}\ket{b}\mapsto \begin{cases}
   \ket{i}\ket{j}\ket{b\oplus 1} & \mbox{if } (v_i,v_j)\in E\\
    \ket{i}\ket{j}\ket{b} & \mbox{otherwise.}
\end{cases}$$
Then for any $S\geq\log^2(n)$, there is a quantum algorithm that decides for any $s,t \in V$ whether there is a directed  path from $s$ to $t$ in $G$ with bounded error using space $O(S)$ and time
$$T\leq 2^{\frac{1}{2}\log n\log\frac{n}{S}+O(\log n\log\log n)}.$$
This algorithm uses $O(\log^2(n))$ quantum space.
\end{theorem}
\begin{proof}
    We analyze the time-space tradeoff of the quantum algorithm of \thm{quantum_dstcon}. Assuming $L \log L = O(n)$, its space complexity becomes $S = O \left(\frac{n}{L} \log n\right)$, with only $O(\log(L)\log(n))=O(\log^2(n))$ quantum space. That is, we can express $L = \Theta \left(\frac{n}{S} \log n \right)$. Substituting this into the time complexity $T = \tO \left( n^{3.5} L^{.5\log(3)-1} (2n+1)^{0.5\log L} \right)$ and taking logarithms, we obtain
    \begin{align*}
        \log T &= 3.5 \log n + \log L \left( \frac{\log (2n + 1)}{2} + .5\log(3)-1 \right) + O(\log\log n)\\
        &= \left( \log \frac{n}{S} + \log \log n + O(1) \right) \left( \frac{\log n}{2}+O(1)\right) + O(\log n)\\
        &= \frac{1}{2} \log (\frac{n}{S}) \log(n) + O \left( \log(n) \log \log(n)\right). & 
    \end{align*}
    \vskip-22pt
\end{proof}
\begin{remark}
    For $S = o(n^{1/2})$, the quantum algorithm of \thm{quantum_tradeoff} achieves a better time-space tradeoff than the classical time-space tradeoff stated in \thm{classical_tradeoff}.
\end{remark}

\section{Quantum short path subroutine}\label{sec:short_path_algorithm}

In this section, we prove \thm{quantum_short_path_complexity} by describing and analyzing a quantum algorithm for $\textsc{Dist}_L(G,s,t)$, for $G=(V,E)$ a directed graph with $V=\{v_1,\dots,v_n\}$, $s,t\in V$ any pair of vertices, and $L\in\mathbb{N}$ a power of 2. The algorithm is designed by exhibiting a switching network (\defin{switching-network}), and then applying \thm{SN_algorithm}. 

In \sec{alg-SN}, we describe the switching network, through a recursive construction. In order to apply \thm{SN_algorithm}, we need to analyze the number of edges in the switching network, and upper bound the distance between its source $\s$ and sink $\t$,
which we do in \sec{alg-wit}; and describe a basis for the space ${\cal B}^\bot$, and a procedure for generating it, which we do in \sec{alg-basis}. Finally, in \sec{final}, we put it all together to prove \thm{quantum_short_path_complexity}.

\subsection{Switching network}\label{sec:alg-SN}

The switching networks we will work with will have vertices represented by a tuple $[u_1,\dots,u_k]\in V^k$, of some number $k$ of vertices of $G$, as well as possibly some additional information. We will not actually care so much about naming conventions for the vertices, but the important detail is that each vertex of a switching network has an associated subset $\{u_1,\dots,u_k\}\subseteq V$ (so the order of the tuple actually doesn't matter). In particular, we will construct switching networks by gluing together switching networks of this form (see \defin{gluing}), and it will be important that any pair of vertices we glue together have the same associated set, so there is no ambiguity. 

The way we construct our switching networks, we will only have an edge between a pair of vertices where the associated sets are of the form $\{u_1,\dots,u_k\}$ and $\{u_1,\dots,u_k,u_{k+1}\}$, and the query label for that edge (see \defin{switching-network}) is $(u_i,u_{k+1})$ for some $i\in [k]$. Such switching networks were first studied in~\cite{potechin2014SwitchingNetworkforDSTCON}. This structure ensures that a vertex $[u]$ can only be connected to a vertex $[u_1,\dots,u_k]$ by a path of ``on'' edges if each $u_i$ is reachable from $u$ in $G$ -- a property that will be crucial for our analysis in \sec{correctness}.

In this section, we will describe and analyze a switching network ${\cal N}_{L}(s)$ of the above described form. This will be built inductively from switching networks ${\cal N}_{2^\ell}(u)$ for $\ell\in\{0,\dots,\log L\}$, and $u\in V$, called the \emph{root}. ${\cal N}_{2^\ell}(u)$ has a single source $[u]$ and $n$ sinks $\{[u,v_i]:v_i\in V\}$. With respect to the $i$-th sink $[u,v_i]$, the switching network computes whether $v_i$ is reachable from $u$ by a path of length at most $2^\ell$, for every $v_i \in V$ simultaneously (i.e. $v_i$ is reachable from $u$ by a path of length at most $2^\ell$ in $G$ if and only if $[ u ]$ and $[ u, v_i ]$ are connected in ${\cal N}_{2^\ell}(u)(G)$). This ``extended boundary'' is used solely for the recursive construction in \sec{graph-construction}. The final construction ${\cal N}_L(s)$ has source $\s=[s]$ and a single sink $\t=[s,t]$. We write ${\cal N}_L(s,t)$ when we want to emphasize this.

\subsubsection{Graph construction}\label{sec:graph-construction}

Define $\Sigma = \{(0,\bar{0})\}\cup \{(1,i):i\in \{0,1\}^{\log n}\}\cup\{(2,j):j\in \{0,1\}^{\log n}\}$, an alphabet of size $2n+1$. We have ensured that all symbols in this alphabet have an obvious representation as a string in $\{0,1,2\}\times \{0,1\}^{\log n}$, but for convenience we will sometimes use $0$ to denote $(0,\bar{0})$, and $1i$ or $2j$ to denote $(1,i)$ or $(2,j)$.
For any $\sigma\in\Sigma^*$, let $|\sigma|$ denote its length, and define:
$$f_1(\sigma):=\left\{\begin{array}{ll}
\max\{i\in \{1,\dots,|\sigma|\}: \sigma_i\in \{1\}\times \{0,1\}^{\log n}\} & \mbox{if }\exists i:\sigma_i\in \{1\}\times \{0,1\}^{\log n}\\
0 & \mbox{else.}
\end{array}\right.$$

\vskip10pt
\noindent We now define ${\cal N}_{2^\ell}(u)$ for $\ell \in [\log L]$, by a recursive construction. 

\paragraph{Base construction.}

For the base case, $\ell = 0$, the switching network ${\cal N}_1(u)$ consists of a source vertex $[ u ]$ connected to $n$ sinks $[ u, v_i ] , v_i \in V$. The edges have query labels $(u, v_i)$, each ``checking" whether there is an edge $(u,v_i)$ in $G$ (see \fig{SN1}). More precisely, we formally define the sets of vertices and edges as follows.
\[
V_1 =  \{ [ u ]\} \cup \{ [ u, v_i ] : v_i \in V \},
\quad\mbox{and}\quad
E_1 = \{\ket{e_i}=\ket{i}:i\in [n]\}. 
\]
Above, we put edge labels in a ket, to emphasize that they form an orthonormal basis of some inner product space.
The incidence of edges and vertices is defined (see \defin{graph} for a reminder of how undirected graphs are specified):
$$E_1^{\rightarrow}([u])=E_1 \quad\mbox{and}\quad E_1^\leftarrow([u])=\emptyset$$
$$\forall i\in [n],\; E_1^{\rightarrow}([u,v_i])=\emptyset \quad\mbox{and}\quad E_1^\leftarrow([u,v_i])=\{\ket{e_i}\}.$$
The query label of the edge $e_i$ is $(u,v_i)$.

\begin{figure}
\centering
\includegraphics[width=0.4\textwidth]{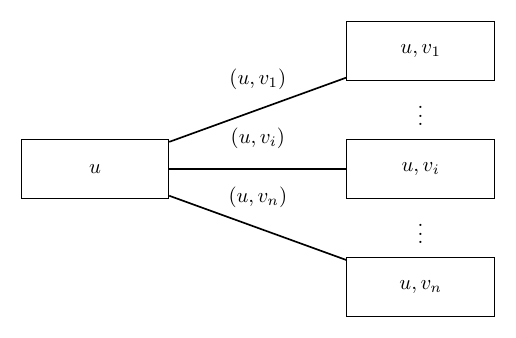}
\caption{Graph construction for the switching network ${\cal N}_1(u)$ that decides whether each vertex of a graph $G = (V,E)$ is reachable from a vertex $u \in V$ by a path of length $1$. Each edge $([u],[u,v_i])$ of ${\cal N}_1(u)$ has query label $(u, v_i)$ and is ``on" in ${\cal N}_1(u)(G)$ if and only if $(u, v_i) \in E$. Therefore, it holds that $v_i$ is reachable from $u$ by a path of length at most $1$ in $G$ if and only if $[u]$ and $[u, v_i]$ are connected in ${\cal N}_1(u)(G)$.}\label{fig:SN1}
\end{figure}

\paragraph{Recursive construction.} Next, we describe the construction of ${\cal N}_{2^\ell}(u)$, assuming that a construction of ${\cal N}_{2^{\ell-1}}(u)$ is given. Let ${\cal N}_{2^{\ell-1}}^{v}(u)$ denote a copy of ${\cal N}_{2^{\ell-1}}(u)$ in which each vertex-tuple of the switching network is augmented with an additional vertex $v \in V$ (although the order in the tuple doesn't matter, for clarity, assume we append $v$ to the front of each tuple).

The construction of ${\cal N}_{2^\ell}(u)$ uses $2n+1$ copies ${\cal N}_{2^{\ell-1}}(u')$ for some $u'$, some of them augmented by additional vertices. Specifically, define:
\begin{align*}
    {\cal N}_{2^{\ell-1}}^0 &= {\cal N}_{2^{\ell-1}}(u)\\
    \forall i\in [n],\; {\cal N}_{2^{\ell-1}}^{(1,i)} &= {\cal N}_{2^{\ell-1}}^u(v_i)\\
    \forall j\in [n],\;{\cal N}_{2^{\ell-1}}^{(2,j)} &= {\sf Rev}({\cal N}_{2^{\ell-1}}^{v_j}(u)).
\end{align*}
Above, we used the notation ${\sf Rev}({\cal N})$ to be the switching network ${\cal N}$ except with the orientation of every edge reversed. 
As we will see shortly, this ensures that all edges of ${\cal N}_{2^\ell}(u)$ have a logical left-to-right orientation. 
Define ${\cal N}_{2^\ell}(u)$ from these $2n+1$ copies of ${\cal N}_{2^{\ell-1}}$ by gluing (as made precise in \defin{gluing}) the $i$-th sink of ${\cal N}_{2^{\ell}}^0$ -- which encodes $[u,v_i]$ -- to the source of ${\cal N}_{2^\ell}^{(1,i)}$ -- which also encodes $[u,v_i]$ -- (for all $i\in [n])$; and gluing the $j$-th sink of ${\cal N}_{2^\ell}^{(1,i)}$ -- which encodes $[u,v_i,v_j]$ -- to the $i$-th sink of ${\cal N}_{2^\ell}^{(2,j)}$ -- which also encodes $[v_j,u,v_i]\equiv [u,v_i,v_j]$ -- (for all $i,j\in [n]$), as in \fig{SNL}. Note that the source of ${\cal N}_{2^\ell}^{(2,j)}$, which encodes $[v_j,u]\equiv [u,v_j]$, is the $j$-th sink of ${\cal N}_{2^\ell}(u)$. The source $[u]$ of ${\cal N}_{2^{\ell-1}}^0$ is the source of ${\cal N}_{2^\ell}(u)$.

\begin{figure}[h]
\centering
\includegraphics[width=\textwidth]{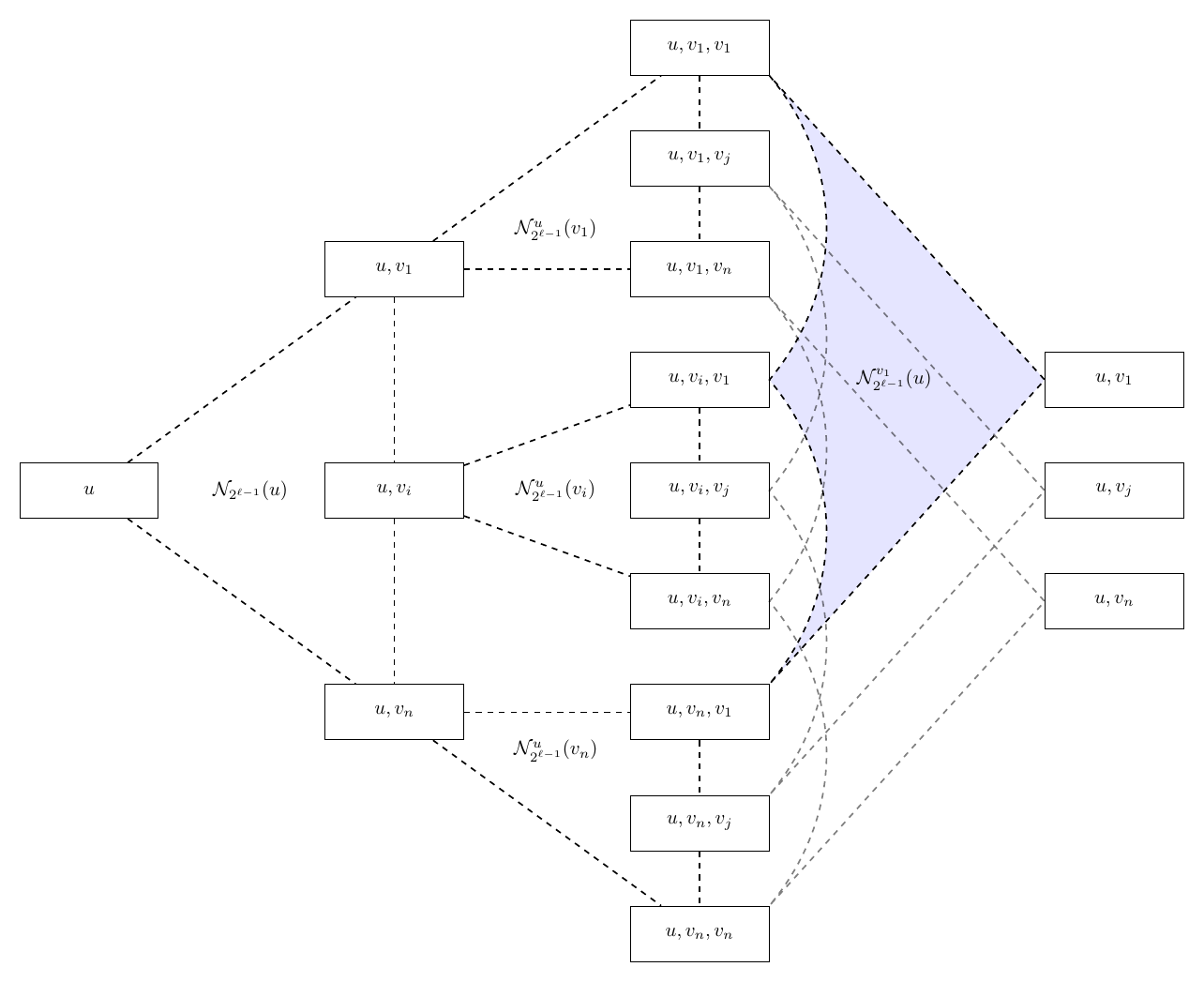}
\caption{Graph construction for the switching network ${\cal N}_{2^\ell}(u)$ that decides whether each vertex of a graph $G = (V,E)$ is reachable from a vertex $u \in V$ by a path of length $2^\ell$.}\label{fig:SNL}
\end{figure}

The edges $E_{2^{\ell}}$ of ${\cal N}_{2^\ell}$ should be the disjoint union of the edge sets of the $2n+1$ copies of ${\cal N}_{2^{\ell-1}}$. We use the elements of $\Sigma$ labeling each copy to make this union disjoint:
\begin{equation}\label{eq:E-2-ell}
E_{2^\ell}=\bigsqcup_{\sigma\in\Sigma}E({\cal N}_{2^{\ell-1}}^\sigma) = \Sigma\times E_{2^{\ell-1}}=\Sigma^{\ell}\times E_1.
\end{equation}
For $\sigma\in \Sigma^\ell$ and $i\in\{0,1\}^{\log n}$, we will sometimes denote the edge $\ket{\sigma,e_i}=\ket{\sigma,i}$ using 
\begin{equation}\label{eq:sign-explainer}
(-1)^{|\sigma|_2}\ket{e_i^\sigma}:=\ket{\sigma,e_i},
\end{equation}
where $|\sigma|_2$ denotes the number of occurrences of $(2,j)$ for some $j$ in $\sigma$. The reason for the sign is that we always want $\ket{\sigma,e_i}$ to represent the $i$-th edge in the $\sigma$-labeled copy of ${\cal N}_1$ \emph{oriented from left-to-right}\footnote{Such an orientation is not strictly necessary, but is more intuitive.} (i.e.~from source towards sinks), and we therefore need to reverse the edge orientations every time we use a copy of ${\cal N}_{2^{\ell}}$ in the $(2,j)$-th position for some $j$. 
Generally, if $\ket{\psi}\in\mathrm{span}\{\ket{e}:e\in E_{2^{\ell}}\}$ -- equivalently, $\psi$ is a function on $E_{2^{\ell}}$ -- and $\sigma\in \Sigma^{\ell'}$, we will let 
\begin{equation}\label{eq:sign-explainer-function}
    \ket{\psi^\sigma}=(-1)^{|\sigma|_2}\ket{\sigma}\ket{\psi},
\end{equation}

which is a state in $\mathrm{span}\{\ket{e}:e\in E_{2^{\ell+\ell'}}\}$ that is only supported on the $\sigma$-labeled copy of ${\cal N}_{2^{\ell}}$ (which we may denote ${\cal N}_{2^{\ell}}^\sigma$) in ${\cal N}_{2^{\ell+\ell'}}$.

Each of the $2n+1$ copies of ${\cal N}_{2^{\ell-1}}$ in ${\cal N}_{2^{\ell}}$ has a unique label $\sigma\in \Sigma$. From this label, and the (global) root $u$ of ${\cal N}_{2^\ell}(u)$, we can extract the root of the specific copy ${\cal N}_{2^{\ell-1}}^\sigma$, and the vertex that is additionally stored in all of its vertices.

Inductively, we assign to each copy of ${\cal N}_{2^{\ell-k}}, k\in [\ell]$ a label $\sigma \in \Sigma^k$, from which we can extract the root of the specific copy ${\cal N}_{2^{\ell-k}}^\sigma$ and the set that is additionally stored in its vertices, knowing the global root $u \in V$. In particular, each copy of ${\cal N}_1$ has a label $\sigma \in \Sigma^{\ell}$ and consists of $n$ edges. We can extract the root of each copy of ${\cal N}_1$ from its label $\sigma$ and the global root $u \in V$ and, hence, recover the edge query labels as well. 
In the following lemma, we show how exactly the query label can be extracted from the label of an edge.

\begin{lemma}\label{lem:edge_encoding}
    Let $\sigma \in \Sigma^{\ell}$ and $i \in [n]$ encode an edge in ${\cal N}_{2^\ell}(u)$. Then its edge is labeled by the query $((v_{\sigma_{\left( f_1(\sigma) \right)_2})},v_i)$, if $f_1(\sigma) \neq 0$, and $(u,v_i)$ otherwise.
\end{lemma}
\begin{proof}
    We prove the statement by induction. In the base case of ${\cal N}_{1}(u)$, the claim is trivial. Since $\sigma \in \Sigma^0$ is an empty string, we have $f_1(\sigma) = 0$.
    The edges are encoded as $\ket{e_i}=\ket{i}$, and the corresponding query labels are $(u,v_i)$. 
    
    For the induction step, assume that the claim holds for ${\cal N}_{2^{\ell-1}}$. The switching network ${\cal N}_{2^\ell}(u)$ consists of $2n+1$ copies of ${\cal N}_{2^{\ell-1}}$ labeled by $0, (1,i), (2,j)$, where $i,j \in [n]$. The edges in each of these switching networks are encoded by $\ket{\tilde\sigma,k}$, where $\tilde\sigma \in \Sigma ^{\ell - 1}$ and $k \in [n]$. First, consider the copies of ${\cal N}_{2^{\ell-1}}$ labeled by $0$ and $(2,j)$ for $j \in [n]$. These are the switching networks ${\cal N}_{2^{\ell-1}}(u)$ and ${\cal N}_{2^{\ell-1}}^{v_j}(u)$ for $j \in [n]$, all of which have $u$ as root, which is the same as the global root of ${\cal N}_{2^\ell}(u)$. By the induction hypothesis, the query label of an edge $\ket{\tilde\sigma, k}$ in one of these switching networks is $((v_{\tilde\sigma_{\left( f_1(\tilde\sigma) \right)_2})},v_k)$, if $f_1(\tilde\sigma) \neq 0$, and $(u,v_k)$ otherwise. In the global switching network ${\cal N}_{2^\ell}(u)$, the same edge is encoded by $\ket{\sigma, k}$, where $\sigma = 0\tilde\sigma$ or $\sigma = (2,j)\tilde\sigma$, depending on the copy of ${\cal N}_{2^{\ell-1}}(u)$. Note that, in both cases, $f_1(\sigma) = f_1(\tilde\sigma)$, which proves the statement for the copies of ${\cal N}_{2^{\ell-1}}$ labeled by $0$ and $(2,j)$ for $j \in [n]$.
    
    Next, consider the copies of ${\cal N}_{2^{\ell-1}}^{(1,i)}$ for $i \in [n]$. These are the switching networks ${\cal N}_{2^{\ell-1}}^u(v_i)$ for  $i \in [n]$ that have $v_i$ as their roots. By the induction hypothesis, the query label of an edge $(\tilde\sigma, k)$ in one of these switching networks is $((v_{\tilde\sigma_{\left( f_1(\tilde\sigma) \right)_2})},v_k)$, if $f_1(\tilde\sigma) \neq 0$, and $(v_i,v_k)$ otherwise. In the global switching network ${\cal N}_{2^\ell}(u)$, the same edge is encoded by $\ket{\sigma, k}$, where $\sigma = (1,i)\tilde\sigma$. Note that $f_1(\sigma) = f_1(\tilde\sigma)$ if $f_1(\tilde\sigma) \neq 0$ and $f_1(\sigma) = 1$ otherwise. This proves the statement for the copies of ${\cal N}_{2^{\ell-1}}$ labeled by $(1,i)$ with $ i \in [n]$, and thus concludes the proof.
\end{proof}

\noindent Finally, we describe a top-down approach to building up ${\cal N}_{2^{\ell+1}}$.
\begin{lemma}\label{lem:SN_top_down}
    Let ${\cal N}'$ be a switching network obtained from ${\cal N}_{2^\ell}$ by replacing each ${\cal N}_1$ block with an ${\cal N}_2$ block with the same boundary. Then ${\cal N}' = {\cal N}_{2^{\ell+1}}$.
\end{lemma}
\begin{proof}
    We prove the statement by induction. The base case is trivial, ${\cal N}_1$ is entirely replaced with ${\cal N}_2$. For the induction step, assume that the statement holds for every ${\cal N}_{2^{\ell '}}$ such that $\ell ' < \ell$. To show it for ${\cal N}_{2^\ell}$, we observe that it consists of ${\cal N}_{2^{\ell - 1}}$ blocks. For each such block, if we replace every ${\cal N}_1$ block with ${\cal N}_2$, it becomes ${\cal N}_{2^\ell}$ by the induction hypothesis. Therefore, the whole switching network becomes ${\cal N}_{2^{\ell + 1}}$. 
\end{proof}

\subsubsection{Correctness of the construction}\label{sec:correctness}
Next, we show that this construction is indeed a switching network that simultaneously decides the connectivity of all $v_i$ to $u$. The proofs in this section follow the general ideas of~\cite[Chapter 3]{potechin2015analyzing}, adapted to the setting of our switching network.

\begin{lemma}\label{lem:SN_correctness}
    For every $\ell\in\{0,\dots,\log L\}$ and $u, v_i \in V$, $v_i$ is reachable from $u$ by a path of length at most $2^{\ell}$ in $G$ if and only if source $[ u ]$ and sink $[ u, v_i ]$ are connected in ${\cal N}_{2^\ell}(u)(G)$.
\end{lemma}

\noindent To prove the statement of \lem{SN_correctness}, we define the following \emph{pebbling game} on the input graph $G$. 

\begin{itemize}
    \item Initially, there is one pebble on the vertex $u$;
    \item For vertices $v, v' \in V$ such that $(v,v')\in E$, if there is a pebble on $v$, it is legal to put a pebble on $v'$ or remove a pebble from $v'$.
\end{itemize}
In such a game, various choices of legal moves give rise to different ``pebblings'' -- sets of vertices containing pebbles -- of the graph. Clearly no vertex not reachable from $u$ can ever be pebbled (i.e., contain a pebble). 
Restricting the number of pebbles available may further restrict the possible configurations achievable, as the following two lemmas show. 

\begin{lemma}[\cite{li1996reversibility}]\label{lem:distance_vs_pebbles}
    Let $D(\ell)$ be the maximal distance from $u$ on which a vertex can be pebbled if $\ell$ pebbles are available in the game. Then $D(\ell) \leq 2^{\ell-1} - 1$.
\end{lemma}
\begin{lemma}\label{lem:pebbles_lower_bound}
    Let $D_r(\ell)$ be the maximal distance from $u$ such that it is possible to obtain pebbling configuration $[u,v]$ for some $v \in V$ using only $\ell$ pebbles. Then $D_r(\ell) \leq 2^{\ell - 2}$
\end{lemma}
\begin{proof}
    We prove the statement by showing that $D_r(\ell + 1) \leq D(\ell) + 1$. The claim then follows from \lem{distance_vs_pebbles}. Assume for contradiction that it is possible to obtain a configuration $[u, u_{D(\ell) + 2}]$ for some $u_{D(\ell) + 2} \in V(G)$ such that the distance from $u$ to $u_{D(\ell) + 2}$ is at least $D(\ell) +2$. Before the last time a pebble is placed on $u_{D(\ell) + 2}$, there is a pebble on some $u_{D(\ell) + 1} \in V(G)$ such that the distance from $u$ to $u_{D(\ell) + 1}$ is at least $D(\ell) + 1$. After this, the pebble is removed from $u_{D(\ell) + 1}$ using only $\ell$ pebbles. Consider the sequence of moves that accomplishes this in reverse. It is a sequence of moves that allows to put a pebble on $u_{D(\ell) + 1}$ using only $\ell$ pebbles, since one pebble always stays on $u_{D(\ell) + 2}$. This is a contradiction, since the distance from $u$ to $u_{D(\ell) + 1}$ is at least $D(\ell) + 1$.
\end{proof}

\begin{lemma}\label{lem:log_path_pebbling}
If a vertex $v \in V(G)$ is reachable from $u$ by a path of length at most $L$ in $G$, then it is possible to obtain pebbling configuration $[u,v]$ using $\log L + 2$ pebbles and $L^{\log 3}$ moves of the pebbling game on $G$.
\end{lemma}
\begin{proof}
    Let $u = u_0, u_1, \ldots ,u_L = v$ be the path from $u$ to $v$ in $G$. Without loss of generality, its length is $L$, which is a power of $2$. We show by induction how to put pebbles on vertices of the path to obtain the configuration $[u,v]$. Assume that there is a sequence of $M_k$ moves that ends in the configuration $[u, u_{k}]$, $k \le L/2$. Then we can obtain $[u, u_{2k}]$ as follows, using $M_{2k}=3T_k$ moves.
    \begin{enumerate}
        \item Perform the sequence of $T_k$ moves to obtain $[u, v_{k}]$;
        \item Perform the same sequence of $T_k$ moves with respect to $v_k$ to obtain $[u, v_{k}, v_{2k}]$;
        \item Perform step 1 in reverse to obtain $[u, v_{2k}]$.
    \end{enumerate}
The base case is putting a pebble on $v_1$ while there is a pebble on $u$, which is a valid pebbling game move. 
In the base case, we use $2$ pebbles while traversing distance $1$ in $M_1=1$ moves. In the induction step, we double the distance from $u$ and use one additional pebble for this. Therefore, we need $\log (L) +2$ pebbles in order to obtain $[u,v]$, and $M_L = 3 M_{L/2} = 3^{\log L} = L^{\log 3}$.
\end{proof}

\begin{lemma}\label{lem:pebbling_only_reachable_vertices}
Let $S$ be a list of vertices of $G$. If pebbling configuration $[S]$ can be obtained in the pebbling game on $G$ then every $v \in S$ is reachable from $u$ in $G$.
\end{lemma}

\begin{proof}
We prove this statement by induction. The base case is the starting configuration $[u]$, the starting vertex $u$ is reachable from itself. For the induction step we assume that $[S]$ is a pebbling configuration that can be obtained in the pebbling game and each vertex in $S$ is reachable from $u$ in $G$. Then, we consider the two types of possible pebbling game moves.
\begin{itemize}
    \item If we remove a pebble from a vertex in $S$, then all remaining vertices remain reachable from~$u$.
    \item If we put a pebble on a new vertex $v'$, then it is reachable by some vertex $v \in S$. Since $v$ is reachable from $u$ by the induction hypothesis, we can combine paths form $u$ to $v$ and from $v$ to $v'$ and conclude that $v'$ is reachable from $u$. \qedhere
\end{itemize}
\end{proof}

\begin{lemma}\label{lem:pebbling_iff_reachablity}
    It is possible to obtain pebbling configuration $[u,v]$ in the pebbling game on $G$ if and only if there is a path from $u$ to $v$ in $G$.
\end{lemma}

\begin{proof}
If there is a path from $u$ to $v$ in $G$, then, by \lem{log_path_pebbling}, configuration $[u,v]$ can be obtained in the pebbling game.

If it is possible to obtain pebbling configuration $[u,v]$ in the pebbling game on $G$, then, by \lem{pebbling_only_reachable_vertices}, $v$ is reachable from $u$ in $G$.
\end{proof}
    
\begin{lemma}\label{lem:SN_path_to_moves}
    Every path in ${\cal N}_{2^\ell}(u)(G)$ corresponds to a sequence of moves in the pebbling game on $G$ transforming the corresponding configurations.
\end{lemma}

\begin{proof}
Every vertex of the switching network contains a list of vertices of $G$ and represents a configuration of pebbles. By construction, every edge that is ``on" is of the form $\{[S], [S,v']\}$ labeled by $(v,v') \in E(G)$ for some $v \in S$. Therefore, depending on the direction, an edge of a path in the switching network corresponds either to adding a pebble to a vertex of $G$ or removing a pebble from a vertex of $G$. 
\end{proof}
    
\begin{proof}[Proof of \lem{SN_correctness}]

We prove by induction that if $v_i$ is reachable from $u$ by a path of length $2^\ell$, then the source $[u]$ and the sink $[u,v_i]$ are connected in $\mathcal{N}_{2^\ell}(u)(G)$.

\paragraph{Base case.} For $\ell=0$, consider $\mathcal{N}_{1}(u)(G)$. If the edge $(u,v_i)$ is present in $G$, then the edge with query label $(u,v_i)$ and endpoints $[u]$ and $[u,v_i]$ is present in $\mathcal{N}_{1}(u)(G)$. Consequently, $[u]$ and $[u,v_i]$ are connected.

\paragraph{Induction step.} Assume the claim holds for $2^{\ell-1}$. Let $p$ be a path of length at most $2^\ell$ from $u$ to $v_i$ in $G$, and let $u'$ denote the midpoint of $p$. Then $u'$ is reachable from $u$ by a path of length at most $2^{\ell-1}$, and $v_i$ is reachable from $u'$ by a path of length at most $2^{\ell-1}$. By the induction hypothesis:
\begin{itemize}
    \item $[u]$ and $[u,u']$ are connected in $\mathcal{N}_{2^{\ell-1}}(u)(G)$;
    \item $[u,u']$ and $[u,u',v_i]$ are connected in $\mathcal{N}^{u}_{2^{\ell-1}}(u')(G)$;
    \item $[u,u',v_i]$ and $[u,v_i]$ are connected in ${\sf Rev}(\mathcal{N}^{v_i}_{2^{\ell-1}}(u)(G))$.
\end{itemize}
Concatenating these paths yields a path from $[u]$ to $[u,v_i]$ in $\mathcal{N}_{2^\ell}(u)(G)$. This completes the induction.

If source $[ u ]$ and sink $[ u, v_i ]$ are connected in ${\cal N}_{2^\ell}(u)(G)$, then, by \lem{SN_path_to_moves}, there is a sequence of pebbling game moves on $G$ that transforms $[u]$ into $[u, v_i]$. Therefore, by \lem{pebbling_iff_reachablity}, there is a path from $u$ to $v_i$ in $G$. Note that the configurations in the sequence contain at most $\ell + 2$ vertices of $G$, by construction of ${\cal N}_{2^\ell}(u)(G)$. Hence, by \lem{pebbles_lower_bound}, the path in $G$ is of length at most~$2^\ell$.    
\end{proof}

\subsection{Complexity analysis of the switching network}\label{sec:alg-wit}

\begin{lemma}\label{lem:positive_witness}
    Fix any $\ell\in\{0,\dots,\log L\}$ and $u,v\in V$ such that there is a path from $u$ to $v$ in $G$ of length at most $2^\ell$.
    Then there is a path connecting the source $[u]$ and the sink $[u,v]$ in ${\cal N}_{2^\ell}(u)(G)$ of length at most $2^{\ell\log 3}$.
\end{lemma}

\begin{proof}
    The statement of this lemma follows directly from the proof of \lem{SN_correctness}.
\end{proof}

\begin{lemma}\label{lem:negative_witness}
    For any $\ell\in\{0,\dots,\log L\}$, and $u\in V$, $|E({\cal N}_{2^\ell}(u))|\leq(2n+1)^{\ell} n$.
\end{lemma}

\begin{proof}$\abs{E({\cal N}_{2^\ell}(u))}  = (2n+1)\abs{E({\cal N}_{2^{\ell - 1}}(u))}
  = (2n+1)^{\ell}\abs{E({\cal N}_{1}(u))} = (2n+1)^{\ell} n.$
\end{proof}

\subsection{Basis of the space of $\s\t$-flows}\label{sec:alg-basis}

In order to apply \thm{SN_algorithm}, we need to describe a basis for ${\cal B}^\bot$, and give an efficient procedure for generating it. By \lem{B_perp}, it is enough to describe a  basis for the $\s\t$-flow space ${\cal F}({\cal N}_{L})$ (see \defin{flow-spaces}), and by \lem{opt_flows_oplus_circulations}, we can further decompose this task into finding a basis for the circulation space ${\cal C}({\cal N}_L)$ (see \defin{flow-spaces}) and an optimal $\s\t$-flow. We first define the basis, and then describe a procedure for generating it.

\subsubsection{Basis definition}

In this section, we prove the following theorem, which, if we take $u=s$ and $v_j=t$, gives us a working basis of ${\cal B}^\bot$ for the switching network ${\cal N}_L(s,t)$.

\begin{theorem}\label{thm:basis_flows_circulations}
    For $j\in \{0,1\}^{\log n}$ and $\ell\in [\log(L)]$, let $\theta_j(2^{\ell-1})$ be the optimal unit $[u], [u,v_j]$-flow  in ${\cal N}_{2^{\ell - 1}}$, and let $\ket{\bar \theta_j(2^{\ell-1})}$ be as in \eq{cropped}. For $i,j \in \{0,1\}^{\log n}$, let 
    \begin{align*}
        \ket{p_{ij}(2^\ell)} = \ket{\bar\theta^0_i(2^{\ell-1})} + \ket{\bar\theta^{1i}_{j}(2^{\ell-1})} - \ket{\bar\theta^{2j}_{i}(2^{\ell-1})}=\ket{0,\bar{0}}\ket{\bar\theta_i(2^{\ell-1})} + \ket{1,i}\ket{\bar\theta_{j}(2^{\ell-1})} + \ket{2,j}\ket{\bar\theta_{i}(2^{\ell-1})},
    \end{align*} 
    where the superscript encodes the copy of ${\cal N}_{2^{\ell - 1}}$ in ${\cal N}_{2^\ell}$ (see~\fig{p_ij}), and for each $x,z\in \{0,1\}^{\log n}$, let  
    \[\ket{\psi_{z,x}(2^{\ell})} = \sum_{j \in \{0,1\}^{\log n}} (-1)^{x \cdot j} \sum_{i \in \{0,1\}^{\log n}} (-1)^{z \cdot i} \ket{p_{ij}(2^{\ell})}.\]
    Finally, let $\ket{\theta_j(L)}$ be as in \eq{flow_state}. Then
    \[\bigcup_{\ell = 1}^{\log(L)}\bigcup_{\sigma \in \Sigma^{\log(L) - \ell}}    \Big\{ \underbrace{\ket{\sigma}\ket{\psi_{z,x}(2^\ell)}}_{=\pm\ket{\psi_{z,x}^\sigma(2^\ell)}} : z,x \in \{0,1\}^{\log n}, z \neq \bar{0} \Big\} \cup \left\{\ket{\theta_{j}(L)}, \ket{[u]}, \ket{[u,v_j]}\right\}\]
    is an orthogonal basis of ${\cal B}^\perp$.
\end{theorem}

This result follows directly from \lem{circulations_basis} (below), which recursively constructs an orthogonal basis for the circulation space, ${\cal C}({\cal N}_{2^\ell})$, of ${\cal N}_{2^\ell}$ (see \defin{flow-spaces}); \lem{optimal_flow_form} (below), which recursively constructs an optimal flow state, which is a basis for the space of optimal flows;
and \lem{opt_flows_oplus_circulations}, which states that the space of flows decomposes as the direct sum of the circulation space and the (one-dimensional) space of optimal flows.

\begin{figure}[h]
\centering
\includegraphics[width=\textwidth]{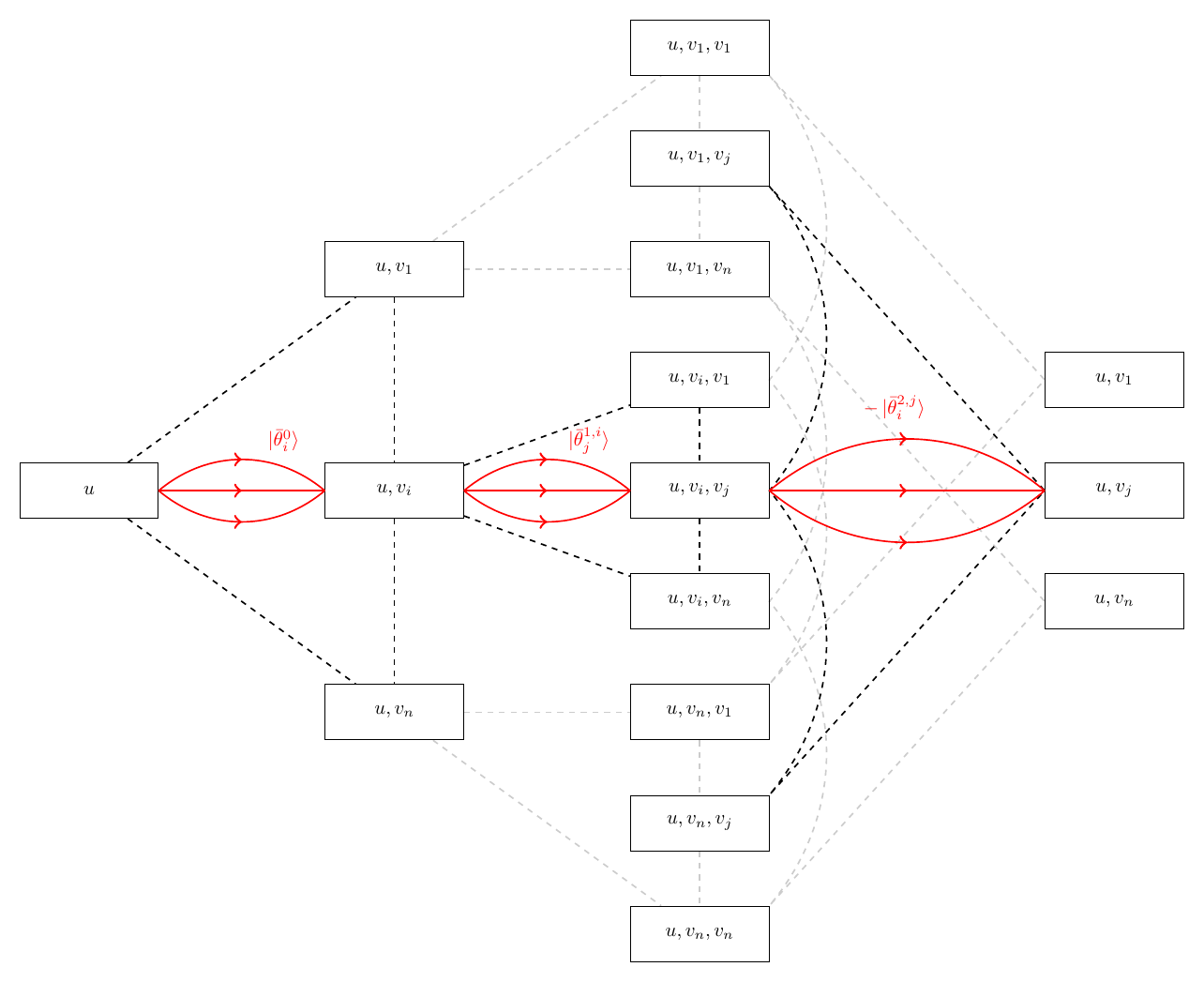}
\caption{Visualization of the state $\ket{p_{ij}} = \ket{\bar\theta^0_i} + \ket{\bar\theta^{1i}_{j}} - \ket{\bar\theta^{2j}_{i}}=\ket{0,\bar{0}}\ket{\bar\theta_i} + \ket{1,i}\ket{\bar\theta_{j}} + \ket{2,j}\ket{\bar\theta_{i}}$ in ${\cal N}_{2^\ell}$, where $j\in \{0,1\}^{\log n}$ and $\ell\in [\log(L)]$. Here $\theta_j$ denotes the optimal unit $[u], [u,v_j]$-flow  in ${\cal N}_{2^{\ell - 1}}$, and $\ket{\bar \theta_j}$ is as in \eq{cropped}. The superscript encodes the copy of ${\cal N}_{2^{\ell - 1}}$ in ${\cal N}_{2^\ell}$.}\label{fig:p_ij}
\end{figure}
\begin{lemma}\label{lem:circulations_basis}
    For $j\in \{0,1\}^{\log n}$ and $\ell\in [\log(L)]$, let $\theta_j$ be the optimal unit $[u], [u,v_j]$-flow  in ${\cal N}_{2^{\ell - 1}}$, and let $\ket{\bar \theta_j}$ be as in \eq{cropped}. For $i,j \in \{0,1\}^{\log n}$, let $\ket{p_{ij}}=\ket{p_{ij}(2^\ell)}$ and $\ket{\psi_{z,x}}=\ket{\psi_{z,x}(2^\ell)}$ be as in \thm{basis_flows_circulations}.     
    Let $\{\ket{b_1},\dots,\ket{b_D}\}$ be an orthogonal basis for $\mathcal{C}({\cal N}_{2^{\ell - 1}})$.
Then
    \[\{\ket{\sigma}\ket{b_d}:\sigma\in\Sigma,d\in [D]\} \cup  \left\{ \ket{\psi_{z,x}} : z,x \in \{0,1\}^{\log n}, z \neq \bar{0} \right\},\]
    is an orthogonal basis of the space of circulations $\mathcal{C}({\cal N}_{2^\ell})$.
\end{lemma}
\begin{proof}
Clearly, the vectors $\{\ket{\sigma}\ket{b_d}\}_{\sigma,d}$ are pairwise orthogonal. 
Moreover, for each $\sigma\in \Sigma$, $\mathrm{span}\{\ket{\sigma}\ket{b_d}:d\in [D]\}$ is contained in $\mathcal{C}({\cal N}_{2^\ell})$, since any circulation on a ${\cal N}_{2^{\ell - 1}}$ block -- in this case, the one labeled by $\sigma$ -- is a circulation on the full graph ${\cal N}_{2^\ell}$. Finally, each state $\ket{\psi_{z,x}}$ is orthogonal to every subspace $\mathrm{span}\{\ket{\sigma}\ket{b_d}:d\in [D]\}$, as it is constructed from optimal flows within the ${\cal N}_{2^{\ell - 1}}$ blocks that are orthogonal to their respective circulation subspaces, by \lem{flow_properties}, and don't overlap other blocks.

Next, we compute the dimension of $\mathcal{C}({\cal N}_{2^\ell})$. By \lem{flow_dimensions}, $\dim \mathcal{C}({\cal N}_{2^\ell}) = \abs{E({\cal N}_{2^\ell})} - \abs{V({\cal N}_{2^\ell})} + 1$. From the construction of ${\cal N}_{2^\ell}$, $\abs{E({\cal N}_{2^\ell})} = (2n+1)\abs{E({\cal N}_{2^{\ell - 1}})}$, and  $\abs{V({\cal N}_{2^\ell})} = (2n+1)\abs{V({\cal N}_{2^{\ell - 1}})} - n^2 - n$. Therefore,
\begin{align*}
\dim \mathcal{C}({\cal N}_{2^\ell}) &= \abs{E({\cal N}_{2^\ell})} - \abs{V({\cal N}_{2^\ell})} + 1\\
&= (2n+1)\abs{E({\cal N}_{2^{\ell - 1}})} - (2n+1) \abs{V({\cal N}_{2^{\ell - 1}})} + n^2 + n + 1\\
&= (2n+1) \underbrace{\left(\abs{E({\cal N}_{2^{\ell - 1}})} - \abs{V({\cal N}_{2^{\ell - 1}})} + 1\right)}_{\dim \mathcal{C}({\cal N}_{2^{\ell - 1}})} +n^2 -n.
\end{align*}
This implies that a basis for $\mathcal{C}({\cal N}_{2^\ell})$ can be constructed by taking orthogonal bases of the circulation spaces of each ${\cal N}_{2^{\ell - 1}}$ block within ${\cal N}_{2^\ell}$, and adding $n^2 - n$ additional states that are orthogonal to these subspace bases.

Next, we argue that $\ket{\psi_{z,x}}$ are indeed circulations. We can interpret $\psi_{z,x}$ as a function on the edges in the natural way: $\psi_{z,x}(e)=\braket{\psi_{z,x}}{e}$. Then, as usual, $\psi_{z,x}(\u)$ denotes the total flow on $\u$: $\psi_{z,x}(\u)=\sum_{e\in E^{\rightarrow}(\u)}\psi_{z,x}(e)-\sum_{e\in E^{\leftarrow}(\u)}\psi_{z,x}(e)$.

Note that each $p_{ij}$ is a $[u],[u,v_j]$-flow in ${\cal N}_{2^\ell}$, as shown in \fig{p_ij}, so in particular, it has boundary $B=\{[u],[u,v_1],\dots,[u,v_n]\}$. 
Hence, each ${\psi_{z,x}}$ is a flow in ${\cal N}_{2^\ell}$ with boundary $B$ as well. Therefore, $\psi_{z,x}(\v) = 0$ for any $\v \in V({\cal N}_{2^\ell}) \setminus B$. It remains to show that $\psi_{z,x}(\v) = 0$ for any $ \v \in B$ as well. Since each $\ket{p_{ij}}$ is a unit $[u],[u,v_j]$-flow and $z \neq \bar{0}$, we get
\begin{align*}
    \psi_{z,x}([u]) &= \sum_{j\in\{0,1\}^n} (-1)^{x \cdot j} \sum_{i\in\{0,1\}^n} (-1)^{z \cdot i}\underbrace{p_{ij}([u])}_{=1} = 0\\
    \psi_{z,x}([u,v_{j'}]) &= \sum_{j\in\{0,1\}^n}(-1)^{x \cdot j} \sum_{i\in\{0,1\}^n} (-1)^{z \cdot i}\underbrace{p_{i,j}([u,v_{j'}])}_{=\delta_{j,j'}} = (-1)^{x \cdot j'} \sum_{i\in\{0,1\}^n} (-1)^{z \cdot i} = 0.
\end{align*}
Hence, each $\ket{\psi_{z,x}}$ is a circulation. To complete the proof, it remains to show that the states $\ket{\psi_{z,x}}$ are pairwise orthogonal, towards which the following claim is helpful. 

\begin{claim}\label{claim:inner-prod-p-psi}
    There exist real numbers $c,c'$ such that for all $i,j,x,z\in\{0,1\}^n$ such that $z\neq\bar{0}$, 
    $$\braket{p_{ij}}{\psi_{z,x}}=(-1)^{z\cdot i}(c(-1)^{x\cdot j}+c'A_{x,j}),$$
    where $A_{x,j}=\sum_{j'\in\{0,1\}^{\log n}:j'\neq j}(-1)^{x\cdot j'}$.
\end{claim}
\begin{proof}
First, we observe that, due to symmetry of the graph ${\cal N}_{2^{\ell - 1}}$, inner products between optimal flows $\ket{\bar\theta_j}$ in this graph are constant. More precisely, there exist $c_0,c_1 \in \mathbb{R}$ such that
\[ \braket{\bar\theta_j}{\bar\theta_{j'}} = \begin{cases}
			c_0, & \text{if $j\neq j'$ }\\
            c_1, & \text{if $j = j'$. } 
		 \end{cases}\]
Therefore, the inner product $\braket{p_{ij}}{p_{i'j'}}$ only depends on whether $i=i'$ and $j=j'$. Indeed,
\begin{equation}\label{eq:braket_p}
    \begin{split}
    \braket{p_{ij}}{p_{i'j'}} &= \braket{\bar\theta^0_i}{\bar\theta^0_{i'}} + \braket{\bar\theta^{1i}_j}{\bar\theta^{1i'}_{j'}} + \braket{\bar\theta^{2j}_i}{\bar\theta^{2j'}_{i'}}\\
    &= \begin{cases}
			c_{00} := c_0 + 0 + 0, & \text{if $i \neq i', j\neq j'$ }\\
            c_{01} := c_0 + 0 + c_0, & \text{if $i \neq i', j = j'$ }\\
            c_{10} := c_1 + c_0 + 0 , & \text{if $i = i', j\neq j'$ }\\
            c_{11} := c_1 + c_1 + c_1, & \text{if $i = i', j = j'$. } 
		 \end{cases}
    \end{split}
\end{equation}
Since $z\neq\bar{0}$, 
$$\sum_{i'\in\{0,1\}^{\log n}:i'\neq i}(-1)^{z\cdot i'} = \sum_{i'\in\{0,1\}^{\log n}}(-1)^{z\cdot i'}-(-1)^{z\cdot i} = - (-1)^{z\cdot i}.$$
We use this to compute:
    \begin{align*}
        \braket{p_{ij}}{\psi_{z,x}} &= \sum_{i',j'\in\{0,1\}^{\log n}}(-1)^{x\cdot j'+z\cdot i'}\braket{p_{ij}}{p_{i'j'}}\\
        &= (-1)^{x\cdot j+z\cdot i}\braket{p_{ij}}{p_{ij}}
        +\sum_{i'\in\{0,1\}^{\log n}:i'\neq i}(-1)^{x\cdot j+z\cdot i'}\braket{p_{ij}}{p_{i'j}}\\
        &\qquad+\sum_{j'\in\{0,1\}^{\log n}:j'\neq j}(-1)^{x\cdot j'+z\cdot i}\braket{p_{ij}}{p_{ij'}}
        +\sum_{i',j'\in\{0,1\}^{\log n}:i'\neq i,j'\neq j}(-1)^{x\cdot j'+z\cdot i'}\braket{p_{ij}}{p_{i'j'}}\\
        &= (-1)^{x\cdot j+z\cdot i}c_{11}
        -(-1)^{x\cdot j+z\cdot i}c_{01}
        +\left((-1)^{z\cdot i}c_{10}
        -(-1)^{z\cdot i}c_{00}\right)\sum_{j'\in\{0,1\}^{\log n}:j'\neq j}(-1)^{x\cdot j'},
    \end{align*}
from which the result follows by taking $c=c_{11}-c_{01}$ and $c'=c_{10}-c_{00}$.
\end{proof}

This allows us to compute the inner product $\braket{\psi_{z,x}}{\psi_{z',x'}}$ with $(z,x) \neq (z',x')$ and $z,z' \neq \bar{0}$:
\begin{align*}
\braket{\psi_{z',x'}}{\psi_{z,x}} &= \sum_{i,j\in\{0,1\}^{\log n}}(-1)^{x'\cdot j+z'\cdot i}\braket{p_{ij}}{\psi_{z,x}}\\
&= \sum_{i,j\in\{0,1\}^{\log n}}(-1)^{x'\cdot j+z'\cdot i}(-1)^{z\cdot i}(c(-1)^{x\cdot j}+c'A_{x,j}) & \text{by \clm{inner-prod-p-psi}}\\
&= \sum_{i\in\{0,1\}^{\log n}}(-1)^{(z'+z)\cdot i} \cdot \sum_{j\in\{0,1\}^{\log n}}(c(-1)^{(x'+x)\cdot j}+c'(-1)^{x'\cdot j}A_{x,j}).
\end{align*}
If $z\neq z'$, then the first product term is 0, and so $\braket{\psi_{z',x'}}{\psi_{z,x}}=0$. Suppose $z=z'$, but $x\neq x'$, which is the only other way to have $(z,x)\neq (z',x')$. Then 
\begin{align*}
\braket{\psi_{z',x'}}{\psi_{z,x}} &= n \cdot \sum_{j\in\{0,1\}^{\log n}}(c(-1)^{(x'+x)\cdot j}+c'(-1)^{x'\cdot j}A_{x,j})
= c'n\sum_{j\in \{0,1\}^{\log n}}(-1)^{x'\cdot j}A_{x,j}.
\end{align*}
Using 
$$A_{x,j}=\sum_{j'\in\{0,1\}^{\log n}:j'\neq j}(-1)^{x\cdot j'}=\left\{\begin{array}{ll}
    n-1 & \mbox{if }x=\bar{0} \\
   -(-1)^{x\cdot j}  & \mbox{else,}
\end{array}\right.$$
we see that if $x\neq \bar{0}$, then
\begin{align*}
\braket{\psi_{z',x'}}{\psi_{z,x}} &= -c'n\sum_{j\in \{0,1\}^{\log n}}(-1)^{(x'+x)\cdot j}=0,
\end{align*}
and otherwise, we must have $x'\neq \bar{0}$, so
\begin{align*}
\braket{\psi_{z',x'}}{\psi_{z,x}} &= c'n(n-1)\sum_{j\in \{0,1\}^{\log n}}(-1)^{x'\cdot j}=0,
\end{align*}
completing the proof.\qedhere
\end{proof}
We now describe the form of optimal flows in ${\cal N}_{2^\ell}$, by recursively combining the optimal flows in ${\cal N}_{2^{\ell-1}}$. For the base case ${\cal N}_1$, the optimal $[u],[u,v_j]$-flow simply assigns a unit of flow to the single edge between $[u]$ and $[u,v_j]$.
\begin{lemma}\label{lem:optimal_flow_form}

For $i,j \in \{0,1\}^{\log n}$, let $\ket{p_{ij}}=\ket{p_{ij}(2^\ell)}$ be as in \thm{basis_flows_circulations}.     
    Then for any $j$, the optimal unit  $[u], [u,v_j]$-flow in ${\cal N}_{2^\ell}(u,v_j)$ is
    \[
    \ket{\theta_{j}(2^\ell)} = -\ket{\leftarrow,[u]} + \frac{1}{n}\sum_{i\in \{0,1\}^{\log n}} \ket{p_{ij}} + \ket{\rightarrow,[u,v_j]}.
    \]   
\end{lemma}
\begin{proof}
Note that each $p_{ij}$ is a $[u],[u,v_j]$-flow in ${\cal N}_{2^\ell}$, as shown in \fig{p_ij}, though it is not optimal, since it doesn't spread out through the $(1, i')$ copies of ${\cal N}_{2^{\ell-1}}$ for $i'\neq i$. Intuitively, $\theta_j(2^\ell)$, which is also easily seen to be a unit $[u],[u,v_j]$-flow on ${\cal N}_{2^\ell}$, is optimal because it spreads the flow across all values of $i$. We prove that it is optimal by showing that $\ket{\theta_{j}(2^\ell)}$ is orthogonal to the circulation subspace $\mathcal{C}({\cal N}_{2^\ell})$. By \lem{flow_properties} and \lem{opt_flows_oplus_circulations}, this orthogonality implies that $\ket{\theta_{j}(2^\ell)}$ is the optimal unit flow.

By \lem{circulations_basis}, the space $\mathcal{C}({\cal N}_{2^\ell})$ has the following orthogonal basis:
\[
\{\ket{\sigma}\ket{b_d}:\sigma\in\Sigma,d\in [D]\}\cup  \left\{ \ket{\psi_{z,x}} : z,x \in \{0,1\}^{\log n}, z \neq \bar{0} \right\},
\]
where $\ket{\psi_{z,x}}=\ket{\psi_{z,x}(2^\ell)}$ is as in \thm{basis_flows_circulations},

and $\{\ket{b_1},\dots,\ket{b_D}\}$ is an orthogonal basis for the circulation space $\mathcal{C}({\cal N}_{2^{\ell - 1}})$.
All $\ket{\theta_{j}(2^\ell)}$ are composed of optimal unit flows of the ${\cal N}_{2^{\ell-1}}$ blocks and boundary states, and therefore 
orthogonal to all vectors $\ket{\sigma}\ket{b_d}$, which are circulations on these blocks. Hence, it remains to show that $\braket{\theta_{j}(2^\ell)}{\psi_{z,x}}=0$ for every $j,z,x \in \{0,1\}^{\log n}$ such that $z \neq \bar 0$. We apply \clm{inner-prod-p-psi} to get:
\begin{equation*}
    n\cdot \braket{\theta_{j}(2^\ell)}{\psi_{z,x}} = \sum_{i \in\{0,1\}^{\log n}} \braket{p_{ij}}{\psi_{z,x}}
    =\sum_{i \in\{0,1\}^{\log n}}(-1)^{z\cdot i} (c(-1)^{x\cdot j}+c'A_{x,j})=0. \qedhere
\end{equation*}
\end{proof}

\subsubsection{Basis Generation}
We show in this section how to prepare the basis from \thm{basis_flows_circulations} in $\tO (1)$ time. All states considered in this section lie in the subspace in which each edge $e \in E({\cal N}_{2^{\ell}})$ is represented by $\frac{1}{\sqrt{2}}\left( \ket{\rightarrow,e} - \ket{\leftarrow,e} \right)$. Without loss of generality, we may therefore describe states in terms of the 
canonical edge states $\{\ket{e} = \ket{\sigma,i} : e \in  E({\cal N}_{2^{\ell}})=\Sigma^\ell\times \{0,1\}^{\log n} \}$, and obtain the isomorphism with the above subspace by appending the auxiliary state $\ket{-} \;=\; \tfrac{1}{\sqrt{2}} \left( \ket{\rightarrow} - \ket{\leftarrow} \right)$ in an additional register.

Recall from \eq{E-2-ell} that edges $E_{2^\ell}$ of ${\cal N}_{2^\ell}$ are labeled $(\sigma,e_i)=(\sigma,i)$ for $\sigma\in\Sigma^{\ell}$ and $i\in\{0,1\}^{\log n}$. For any $\sigma\in\Sigma^\ell$, letting $\bar\sigma\in\{0,1,2\}^\ell$ be such that for all $t\in [\ell]$, $\sigma_t=(\bar{\sigma}_t,i)$ for some $i\in\{0,1\}^{\log n}$ (that is, $\bar\sigma$ encodes the first part of each entry of $\sigma$), we can break $E_{2^\ell}$ into $3^{\ell}$ \emph{layers}, defined, for each $\tau\in \{0,1,2\}^\ell$:
\begin{equation}\label{eq:E-sigma}
E_{\tau}:=\{(\sigma,i)\in E_{2^{\ell}}:\bar\sigma=\tau\}
=\{(\sigma,i)\in \Sigma^\ell\times \{0,1\}^{\log n}:\bar\sigma=\tau\}.
\end{equation}
A crucial step in our basis generation subroutine will be taking uniform superpositions over these layers. We first prove two lemmas about the size of these layers that will enable this.

\begin{lemma}\label{lem:layer_size} Fix $\ell\in [\log(L)]$.
    Let $\tau \in \{0,1,2\}^{\ell}$ encode a layer $E_{\tau}$ of edges in ${\cal N}_{2^\ell}$. Then $\abs{E_{\tau}} = n^{1 + \abs{\tau} - \abs{\tau}_0}$, where $\abs{\tau}_0$ denotes the number of zeros in $\tau$.
\end{lemma}
\begin{proof}
We prove the statement by induction on $\ell$. For the base case, let $\ell=1$, so we have for any $\tau\in\{0,1,2\}$:
$$|E_\tau|=|\{(\sigma,i)\in \Sigma\times\{0,1\}^{\log n}:\bar{\sigma}=\tau\}|=n|\{\sigma\in \Sigma:\bar{\sigma}=\tau\}|.$$
If $\tau=0$, the only $\sigma\in \Sigma=\{(0,\bar{0}),(1,i),(2,j):i,j\in\{0,1\}^{\log n}\}$ such that $\bar{\sigma}=\tau$ is $\sigma=(0,\bar{0})$, and so
$$|E_\tau|=n\cdot 1=n^{1+|\tau|-|\tau|_0}$$
since $|\tau|=|\tau|_0=1$. Otherwise, $\tau\in\{1,2\}$, and we have:
$$|E_\tau|=n|\{(\tau,i):i\in\{0,1\}^{\log n}\}|=n^2=n^{1+|\tau|-|\tau|_0}$$
since $|\tau|=1$ and $|\tau|_0=0$. 

For the induction step, assume $\ell>1$. 
Referring to \eq{E-sigma}, we have, for any $\tau\in\{0,1,2\}^{\ell-1}$:
\begin{equation}\label{eq:E-split}
\begin{split}
    E_{0\tau} &= \{((0,\bar{0})\sigma,i)\in E_{2^{\ell}}:\bar{\sigma}=\tau\}
=\{((0,\bar{0}),e):e\in E_{\tau}\}\\
\mbox{and for $a\in\{1,2\}$, }E_{a\tau} &= \{((a,j)\sigma,i)\in E_{2^{\ell}}:\bar{\sigma}=\tau\}
=\{((a,j),e):e\in E_{\tau},j\in\{0,1\}^{\log n}\}.
\end{split}
\end{equation}
From this, and the induction hypothesis, we have:
\begin{equation*}
\begin{split}
    |E_{0\tau}| &= |E_{\tau}| = n^{1+|\tau|-|\tau|_0}=n^{1+|\tau|+1-(|\tau|_0+1)}=n^{1+|0\tau|-|0\tau|_0}\\
\mbox{and for $a\in\{1,2\}$, }|E_{a\tau}| &= n|E_\tau|= n\cdot n^{1+|\tau|-|\tau|_0} = n^{1+(1+|\tau|)-|\tau|_0} = n^{1+|a\tau|-|a\tau|_0}.
\end{split}
\end{equation*}
Thus, for any $\tau'\in \{0,1,2\}^{\ell}$, we can conclude $|E_{\tau'}|=n^{1+|\tau'|-|\tau'|_0}$.
\end{proof}

\begin{lemma}\label{lem:sum_of_level_sizes}
    For all $\ell\in [\log(L)]$, $\sum_{\tau\in\{0,1,2\}^{\ell}}\frac{1}{|E_\tau|}=\frac{(n+2)^{\ell}}{n^{\ell+1}}=\frac{1}{n}\left(1+\frac{2}{n}\right)^{\ell}$.
\end{lemma}
\begin{proof}
We first use \lem{layer_size} to compute:
    \begin{align*}
        \sum_{\tau\in\{0,1,2\}^{\ell}}\frac{1}{|E_\tau|} &= \sum_{\tau\in\{0,1,2\}^{\ell}}\frac{1}{n^{1+|\tau|-|\tau|_0}}
        =\sum_{t=0}^{\ell}\binom{\ell}{t}2^{\ell-t}n^{t-1-\ell}.
    \end{align*}
    Above, we used the fact that for any $t\in\{0,\dots,\ell\}$, there are $\binom{\ell}{t}2^{\ell-t}$ strings in $\{0,1,2\}^{\ell}$ with exactly $t$ 0s. Continuing, we have:
    \begin{align*}
        \sum_{\tau\in\{0,1,2\}^{\ell}}\frac{1}{|E_\tau|} 
        &=\frac{1}{n^{\ell+1}}\sum_{t=0}^{\ell}\binom{\ell}{t}2^{\ell-t}n^{t}
        = \frac{1}{n^{\ell+1}}(n+2)^{\ell},
    \end{align*}
    by the binomial theorem. The result follows.
\end{proof}

Next, we describe a subroutine for generating superpositions over all optimal flows, which will be a key subroutine in our basis generation.

\begin{lemma}\label{lem:sum_of_flows}
For all $\ell \in [\log(L)]$, a map that acts as
$$\ket{0}\mapsto \propto \sum_{i\in \{0,1\}^{\log n}}\ket{\bar\theta_i(2^\ell)}$$
can be implemented in $\tO(1)$ steps.
\end{lemma}
\begin{proof}
We first show by induction that 
\begin{equation}\label{eq:bar-theta_j-eqs}
\sum_{j\in\{0,1\}^{\log n}}\ket{\bar\theta_j(2^\ell)} = n\cdot \sum_{\tau \in \{0,1,2\}^{\ell}} \frac{1}{\abs{E_{\tau}}} \sum_{e \in E_{\tau}} \ket{e}.
\end{equation}
Using the definition in the statement of \thm{basis_flows_circulations}, we can write
\[
\sum_{j\in\{0,1\}^{\log n}}^n\ket{\bar\theta_j(2^\ell)} = \frac{1}{n} \sum_{j\in \{0,1\}^{\log n}} \sum_{i\in\{0,1\}^{\log n}} \ket{p_{ij}(2^\ell)}.
\]
In the base case of ${\cal N}_{2}$ ($\ell=1$), referring to the definition of $\ket{p_{ij}}$ in \thm{basis_flows_circulations}, we have 
\begin{align*}
    \ket{p_{ij}(2)} &= \ket{\bar{\theta}_i^0(1)}+\ket{\bar\theta_j^{1i}(1)}-\ket{\bar\theta_i^{2j}(1)}
    = \ket{0,\bar{0}}\ket{\bar\theta_i}+\ket{1,i}\ket{\bar\theta_j}+\ket{2,j}\ket{\bar\theta_i}\\
    &= \ket{0,\bar{0}}\ket{e_i}+\ket{1,i}\ket{e_j}+\ket{2,j}\ket{e_i}
\end{align*}
by \eq{sign-explainer-function}, and the fact that an optimal $[u],[u,v_i]$-flow in ${\cal N}_1$ simply assigns a unit of flow to $e_i$. 
$$\ket{p_{ij}(2)}=\ket{e_i^0}+\ket{e_j^{1i}}+\ket{e_i^{2j}}.$$
We thus get
\begin{align*}
    \sum_{j\in\{0,1\}^{\log n}}\ket{\bar\theta_j(2)} &= \frac{1}{n} \sum_{j\in\{0,1\}^{\log n}} \sum_{i\in\{0,1\}^{\log n}} \ket{p_{ij}(2)}\\
    &= n \cdot \frac{1}{n} \sum_{i\in\{0,1\}^{\log n}} \ket{0,\bar{0},e_i} +  n \cdot \frac{1}{n^2} \sum_{i,j\in\{0,1\}^{\log n}} \ket{1,i,e_j} + n \cdot \frac{1}{n^2} \sum_{i,j\in\{0,1\}^{\log n}} \ket{2,j,e_i}.
\end{align*}
Since the layer encoded by $0$ consists of $n$ edges $E_{0}=\{((0,\bar{0}),e_i): i \in \{0,1\}^{\log n}\}$, and the layers encoded by $1$ and $2$ consist of $n^2$ edges $E_{1}=\{{((1,i),e_j)}: i, j \in\{0,1\}^{\log n}\}$ and $E_2=\{{((2,j),e_i)}: i, j \in\{0,1\}^{\log n}\}$ respectively, we conclude that the equality holds in the base case.

Assume, for the induction step, that the equality from \eq{bar-theta_j-eqs} holds for ${\cal N}_{2^{\ell-1}}$. We have
\begin{multline*}
    \sum_{j\in\{0,1\}^{\log n}}\ket{\bar\theta_j(2^\ell)} = \frac{1}{n} \!\!\sum_{i,j\in\{0,1\}^{\log n}} \ket{p_{ij}(2^\ell)}
    = \frac{1}{n}\!\! \sum_{i,j\in\{0,1\}^{\log n}}  \left( \ket{\bar\theta^0_i(2^{\ell-1})} + \ket{\bar\theta^{1i}_j(2^{\ell-1})} - \ket{\bar\theta^{2j}_i(2^{\ell-1})} \right)\\
    = \sum_{i\in\{0,1\}^{\log n}} \ket{0,\bar{0}}\ket{\bar\theta_i(2^{\ell-1})} + \frac{1}{n} \sum_{i,j\in\{0,1\}^{\log n}} \ket{1,i}\ket{\bar\theta_j(2^{\ell-1})} + \frac{1}{n} \sum_{i,j\in\{0,1\}^{\log n}}\ket{2,j}\ket{\bar\theta_i(2^{\ell-1})},
\end{multline*}
where we again used \eq{sign-explainer-function}. 
We can use the induction hypothesis on each of the three terms to get:
\begin{multline*}
    \sum_{j\in\{0,1\}^{\log n}}\ket{\bar\theta_j(2^\ell)}
    = n\ket{0,\bar{0}} \sum_{\tau \in \{0,1,2\}^{\ell - 1}} \frac{1}{\abs{E_{\tau}}} \sum_{e \in E_{\tau}} \ket{e}
    + \sum_{i\in\{0,1\}^{\log n}} \ket{1,i} \sum_{\tau \in \{0,1,2\}^{\ell - 1}} \frac{1}{\abs{E_{\tau}}} \sum_{e \in E_{\tau}} \ket{e} \\
    + \sum_{j\in\{0,1\}^{\log n}} \ket{2,j} \sum_{\tau \in \{0,1,2\}^{\ell - 1}} \frac{1}{\abs{E_{\tau}}} \sum_{e \in E_{\tau}} \ket{e}.
\end{multline*}
From \eq{E-split}, and \lem{layer_size}, it follows that 
\begin{align*}
    \sum_{j\in\{0,1\}^{\log n}}\ket{\bar\theta_j(2^\ell)}
    &= \sum_{\tau \in \{0,1,2\}^{\ell - 1}} \frac{1}{n^{|\tau|-|\tau|_0}} \sum_{e \in E_{\tau}} \ket{(0,\bar{0}),e}
    \\
    &\qquad\qquad\qquad\qquad\qquad\quad+\sum_{a\in\{1,2\}}\sum_{\tau \in \{0,1,2\}^{\ell - 1}} \frac{1}{n^{1+|\tau|-|\tau|_0}} \sum_{i\in\{0,1\}^{\log n},e \in E_{\tau}} \ket{(a,i),e}\\
    &= \sum_{\tau \in \{0,1,2\}^{\ell - 1}} \frac{1}{n^{|0\tau|-|0\tau|_0}} \sum_{e' \in E_{0\tau}} \ket{e'}
    + \sum_{a\in\{1,2\}}\sum_{\tau \in \{0,1,2\}^{\ell - 1}} \frac{1}{n^{|a\tau|-|a\tau|_0}} \sum_{e' \in E_{a\tau}} \ket{e'} \\
    &= n \sum_{\tau'\in \{0,1,2\}^{\ell}}\frac{1}{|E_{\tau'}|}\sum_{e'\in E_{\tau'}}\ket{e'},
\end{align*}
as desired.
Next, we show how to prepare this state in two steps.

\paragraph{Step 1.} Prepare a superposition of layer names with the correct amplitudes:
\[
\frac{1}{\sqrt{\sum_{\tau' \in \{0,1,2\}^{\ell}} \frac{1}{\abs{E_{\tau'}}}}} \sum_{\tau \in \{0,1,2\}^{\ell}} \frac{1}{\sqrt{\abs{E_{\tau}}}} \ket{\tau}.
\]
By \rem{superposition_preparation_012}, this superposition can be prepared in $\tO(1)$ time if the amplitudes $$\frac{1}{\sqrt{\sum_{\tau' \in \{0,1,2\}^{\ell}} \frac{1}{\abs{E_{\tau'}}}}} \times \frac{1}{\sqrt{\abs{E_{\tau}}}}$$ and partial sums $$S(p) = \sum_{\tau:\tau \text{ has prefix } p} \frac{1}{\abs{E_{\tau}}}$$ can be computed in $\tO (1)$ time. 
We find closed-form expressions for these quantities, which shows that they can indeed be evaluated efficiently. We start by applying \lem{layer_size} and \lem{sum_of_level_sizes} to compute the amplitudes:
$$\frac{1}{\sqrt{\sum_{\tau' \in \{0,1,2\}^{\ell}} \frac{1}{\abs{E_{\tau'}}}}} \frac{1}{\sqrt{\abs{E_{\tau}}}} = \left( \frac{(n+2)^{\ell}}{n^{\ell+1}}\right)^{-1/2} \left( n^{1 + \ell - \abs{\tau}_0} \right)^{-1/2}
= \sqrt{\frac{n^{|\tau|_0}}{(n+2)^{\ell}}}.$$
Finally, we compute the partial sums, again using \lem{layer_size} and \lem{sum_of_level_sizes}. Let $p \in \{ 0,1,2 \}^k$ for $k \in [\ell]$ be a fixed prefix. Then 
\begin{align*}
    S(p) &= \sum_{\tau:\tau \text{ has prefix } p} \frac{1}{\abs{E_{\tau}}}
    = \sum_{\tau' \in \{ 0,1,2 \}^{\ell -k}} \frac{1}{\abs{E_{p \tau'}}}
    = \sum_{\tau' \in \{ 0,1,2 \}^{\ell -k}} \frac{1}{n^{1 + \ell - (\abs{p}_0 + \abs{\tau'}_0)}}\\
    &= \frac{1}{n^{k - \abs{p}_0}} \sum_{\tau' \in \{ 0,1,2 \}^{\ell -k}} \frac{1}{n^{1 + (\ell - k) - \abs{\tau'}_0}}
    = \frac{1}{n^{k - \abs{p}_0}} \sum_{\tau' \in \{ 0,1,2 \}^{\ell -k}} \frac{1}{\abs{E_{\tau'}}}\\
    &= \frac{1}{n^{k - \abs{p}_0 +1}} \left(1+\frac{2}{n}\right)^{\ell - k}
    = \frac{(n+2)^{\ell-k}}{n^{\ell+1-|p|_0}}.
\end{align*}

\paragraph{Step 2.} Map each $\ket{\tau}$ to the uniform superposition over the edges in the layer $E_{\tau}$.  
Since these edges are encoded as $\ket{\sigma, i}$ with $\sigma \in \Sigma^{\ell}$, $i \in \{0,1\}^{\log n}$, and $\bar{\sigma} = \tau$, this mapping can be performed in $\tO(1)$ time, as follows. For each $k\in [\ell]$, if $\tau_k\in\{1,2\}$, generate a uniform superposition over $j\in\{0,1\}^{\log n}$, to get a superposition over $\sigma_k$ such that $\bar{\sigma}_k=\tau_k$. Finally, generate a uniform superposition over $i\in\{0,1\}^{\log n}$. Letting $S(\tau_k)=\{\bar{0}\}$ if $\tau_k=0$ and $\{0,1\}^{\log n}$ otherwise, this mapping acts as:
\begin{align*}
\ket{\tau}=\ket{\tau_1}\otimes\dots\otimes\ket{\tau_\ell}
&\mapsto \left(\bigotimes_{k=1}^\ell \ket{\tau_k}\sum_{j\in S(\tau_k)}\frac{1}{\sqrt{|S(\tau_k)|}}\ket{j} \right) \otimes \sum_{i\in\{0,1\}^{\log n}}\frac{1}{\sqrt{n}}\ket{i} \\
&= \frac{1}{\sqrt{n^{1+\ell-|\tau|_0}}}\sum_{e\in E_\tau}\ket{e} = \frac{1}{\sqrt{|E_\tau|}}\sum_{e\in E_\tau}\ket{e}.
\end{align*}
We thus end up with a state proportional to 
\begin{equation}\label{eq:sum_of_flows_normalized}
     \sum_{\tau \in \{0,1,2\}^{\ell}} \frac{1}{\sqrt{\abs{E_{\tau}}}} \sum_{e\in E_{\tau}} \frac{1}{\sqrt{\abs{E_{\tau}}}} \ket{e}
        = \sum_{\tau \in \{0,1,2\}^{\ell}} \frac{1}{\abs{E_{\tau}}} \sum_{e\in E_{\tau}}  \ket{e},
\end{equation}
which is proportional to $\sum_{j\in\{0,1\}^{\log n}}\ket{\bar\theta_j(2^\ell)}$
by \eq{bar-theta_j-eqs}.
\end{proof}

Building on the previous lemma, we describe how to generate states that will shortly be shown to make up part of the states $\ket{\psi_{z,x}(2^\ell)}$ from the basis in \thm{basis_flows_circulations}.

\begin{lemma}\label{lem:Fourier_flows}
    For all $\ell > 0$, there is a circuit $C_{2^\ell}$ that acts, for all $x\in \{0,1\}^{\log n}$, as
    $$\ket{x}\mapsto \propto \sum_{j\in\{0,1\}^{\log n}}(-1)^{x\cdot j}\ket{\bar\theta_j(2^\ell)}$$
    and uses $\tO(1)$ gates.
\end{lemma}
\begin{proof}
We start by noticing that if $x = \bar 0$ then the desired state is $\sum_{j\in\{0,1\}^{\log n}}\ket{\bar\theta_j(2^\ell)}$. By \lem{sum_of_flows}, a state proportional to this one can be prepared in $\tO (1)$. Hence, we assume that $x \neq \bar 0$ for the rest of the proof. We prove the statement by induction. Consider ${\cal N}_{2}$ for the base case. For a fixed $x \neq \bar 0$, we can write the following.
\begin{equation}\label{eq:Fourier_flows_base}
\begin{split}
    \sum_{j\in\{0,1\}^{\log n}}\!\!\!\!(-1)^{x\cdot j}\ket{\bar\theta_j(2)} &= \sum_{j\in\{0,1\}^{\log n}}\!\!\!\!(-1)^{x\cdot j} \sum_{i\in\{0,1\}^{\log n}} \left(\ket{e^0_i} + \ket{e^{1i}_j} - \ket{e^{2j}_i}\right)\\
    &= \underbrace{\sum_{j\in\{0,1\}^{\log n}}\!\!\!\!(-1)^{x\cdot j}}_{=0} \sum_{i\in\{0,1\}^{\log n}}\!\!\!\! \ket{e^0_i} +\!\!\!\! \sum_{i,j\in\{0,1\}^{\log n}}\!\!\!\! (-1)^{x\cdot j} \ket{e^{1i}_j} -\!\!\!\! \sum_{i,j\in\{0,1\}^{\log n}} \!\!\!\!(-1)^{x\cdot j} \ket{e^{2j}_i}\\
    &= \sum_{i,j\in\{0,1\}^{\log n}} \!\!\!\!(-1)^{x\cdot j} \ket{1,i}\ket{j} + \sum_{i,j\in\{0,1\}^{\log n}}  \!\!\!\!(-1)^{x\cdot j} \ket{2,j}\ket{i}\\
    &\propto \ket{1}H^{\otimes \log n}\ket{\bar{0}}H^{\otimes \log n}\ket{x}+\ket{2} H^{\otimes \log n}\ket{x}H^{\otimes \log n}\ket{\bar{0}}.
\end{split}
\end{equation}
We used the fact that edges are encoded as $\ket{e^{1i}_j} = \ket{1, i}\ket{j}$ and $\ket{e^{2j}_i} = -\ket{2, j}\ket{i}$ (see \eq{sign-explainer}).
Thus, to prepare this state from $\ket{x}$, first make a uniform superposition over $\ket{1}$ and $\ket{2}$, and then swap the second and third register controlled on the first register:
$$\ket{0,\bar{0}}\ket{x} \mapsto \frac{1}{\sqrt{2}}\ket{1,\bar{0}}\ket{x}+\frac{1}{\sqrt{2}}\ket{2,\bar{0}}\ket{x}
\mapsto \frac{1}{\sqrt{2}}\ket{1}\ket{\bar{0}}\ket{x}+\frac{1}{\sqrt{2}}\ket{2}\ket{x}\ket{\bar{0}}.$$
We can complete the computation by applying Hadamards to all but the first register. This can be done in $O(\log n)$ gates. 

Assume, for the induction step, that there is a circuit $C_{2^{\ell-1}}$ with $\tO(1)$ gates that implements 
$$\ket{x}\mapsto \frac{\sum_{j\in\{0,1\}^{\log n}}(-1)^{x\cdot j}\ket{\bar\theta_j(2^{\ell-1})}}{\norm{\sum_{j\in\{0,1\}^{\log n}}(-1)^{x\cdot j}\ket{\bar\theta_j(2^{\ell-1})}}}$$
for any $x$.
We rewrite the desired state analogously to the base case (all sums below are over $\{0,1\}^{\log n}$).
\begin{equation*}
\begin{split}
    \sum_{j\in\{0,1\}^{\log n}}(-1)^{x\cdot j}\ket{\bar\theta_j(2^\ell)} &= \sum_{j}(-1)^{x\cdot j} \sum_{i} \left(\ket{\bar\theta^0_i(2^{\ell-1})} + \ket{\bar\theta^{1i}_j(2^{\ell-1})} - \ket{\bar\theta^{2j}_i(2^{\ell-1})}\right)\\
    &= \underbrace{\sum_{j}(-1)^{x\cdot j}}_{=0} \sum_{i} \ket{\bar\theta^0_i(2^{\ell-1})} + \sum_{j}\sum_i (-1)^{x\cdot j} \ket{\bar\theta^{1i}_j(2^{\ell-1})} - \sum_{j}\sum_i (-1)^{x\cdot j} \ket{\bar\theta^{2j}_i(2^{\ell-1})}\\
    &= \sum_{j} \sum_{i} (-1)^{x\cdot j} \ket{1,i}\ket{\bar\theta_j(2^{\ell-1})} + \sum_{j} \sum_{i} (-1)^{x\cdot j}  \ket{2,j}\ket{\bar\theta_i(2^{\ell-1})}.
\end{split}
\end{equation*}
We used $\ket{\bar\theta_j^{1i}(2^{\ell-1})}=\ket{1,i}\ket{\bar\theta_j(2^{\ell-1})}$ and $\ket{\bar\theta_i^{2j}(2^{\ell-1})}=\ket{2,j}\ket{\bar\theta_i(2^{\ell-1})}$, by \eq{sign-explainer}. Continuing, we use the induction hypothesis to compute (again, all sums are over $\{0,1\}^{\log n}$):
\begin{equation}\label{eq:theta-j-sum-recurse}
\begin{split}
\sum_{j\in\{0,1\}^{\log n}}^n(-1)^{x\cdot j}\ket{\bar\theta_j(2^\ell)} &= \sum_{i}\ket{1,i}\otimes \sum_{j} (-1)^{x\cdot j}\ket{\bar\theta_j(2^{\ell-1})} + \sum_{j} (-1)^{x\cdot j}\ket{2,j}\otimes \sum_{i}\ket{\bar\theta_i(2^{\ell-1})}\\
    &=\ket{1}\otimes \sqrt{n}  H^{\otimes \log n}\ket{\bar{0}} \otimes \norm{\sum_{j\in\{0,1\}^{\log n}}(-1)^{x\cdot j}\ket{\bar\theta_j(2^{\ell-1})}} C_{2^{\ell-1}}\ket{x}\\
    & \qquad\qquad +\ket{2}\otimes \sqrt{n}H^{\otimes \log n}\ket{x}\otimes \norm{\sum_{j\in\{0,1\}^{\log n}}\ket{\bar\theta_j(2^{\ell-1})}} C_{2^{\ell-1}}\ket{\bar{0}}.
\end{split}\end{equation}

Thus, as in the base case, we first put appropriate weight on $\ket{1}$ and $\ket{2}$, using a simple rotation:
$$\ket{0,\bar{0}}\ket{x}\mapsto \left(\sqrt{\alpha}\ket{1}+ \sqrt{1-\alpha}\ket{2}\right)\ket{\bar{0}}\ket{x},$$
where 
$$\alpha = \frac{\norm{\sum_{j\in\{0,1\}^{\log n}} (-1)^{x\cdot j} \ket{\bar\theta_j(2^{\ell-1})}}^2}{\norm{\sum_{j\in\{0,1\}^{\log n}} (-1)^{x\cdot j} \ket{\bar\theta_j(2^{\ell-1})}}^2 + \norm{\sum_{i\in\{0,1\}^{\log n}} \ket{\bar\theta_i(2^{\ell-1})}}^2}.$$
We will describe how to do this first rotation shortly.
As in the base step, we next swap the second and third registered controlled on the first, and then complete the computation by applying $H^{\otimes \log n}$ to the second register and $C_{2^{\ell-1}}$ to the third register.

We complete the proof by describing how to do the first rotation in terms of $\alpha$. By \lem{superposition_preparation}, such a superposition can be prepared in $\tO (1)$ time if $\alpha$ can be computed in $\tO (1)$ time. Hence, it suffices to compute closed-form expressions for the two norms in the definition of $\alpha$. From \eq{bar-theta_j-eqs} and \lem{sum_of_level_sizes}, we can conclude that for any $\ell>0$:
\begin{equation}\label{eq:sum_of_flows_norm}
    N_{\bar{0}}(2^\ell)\coloneq\norm{\sum_{i\in\{0,1\}^{\log n}} \ket{\bar\theta_i(2^\ell)}}^2 = n^2\sum_{\tau\in\{0,1,2\}^\ell}\frac{1}{|E_\tau|}
    = n^2\frac{(n+2)^\ell}{n^{\ell+1}}.
\end{equation}
Letting $N_x(2^\ell) \coloneq \norm{\sum_{j\in\{0,1\}^{\log n}} (-1)^{x\cdot j} \ket{\bar\theta_j(2^\ell)}}^2$, we show by induction that for $\ell\geq 1$:
\begin{equation}\label{eq:N_xL}
    N_x(2^\ell) = n^{\ell +1} + \frac{n^3}{(n-2)(n+1)} \left(n^{\ell}-\frac{(n+2)^\ell}{n^\ell}\right).
\end{equation}
For the base step, we can observe from \eq{Fourier_flows_base} that $N_x(2)=2n^2$, which is equal to the right-hand-side of \eq{N_xL} when $\ell=1$. For the induction step, we can use \eq{theta-j-sum-recurse} to show that for $\ell>1$:
\begin{equation*}\begin{split}
    N_x(2^\ell) &= \norm{\sqrt{n}H^{\otimes\log n}\ket{\bar{0}}\otimes \sqrt{N_x(2^{\ell-1})}C_{2^{\ell-1}}\ket{x}}^2+\norm{\sqrt{n}H^{\otimes\log n}\ket{x}\otimes \sqrt{N_{\bar{0}}(2^{\ell-1})}C_{2^{\ell-1}}\ket{\bar{0}}}^2\\
    &=n N_x(2^{\ell-1})+n N_{\bar{0}}(2^{\ell-1})\\
    &= n\left(n^{\ell}+\frac{n^3}{(n-2)(n+1)} \left(n^{\ell-1}-\frac{(n+2)^{\ell-1}}{n^{\ell-1}}\right)\right)+n^2\frac{(n+2)^{\ell-1}}{n^{\ell-1}},
\end{split}\end{equation*}
by induction. This is easily shown to be equal to the right-hand-side of \eq{N_xL}.
\end{proof}

We now describe how to generate the states $\ket{\psi_{z,x}(2^\ell)}$, which form most of the basis described in \thm{basis_flows_circulations}.
\begin{lemma}\label{lem:circ_preparation}
    For all $\ell \in [\log (L)]$, a map that acts, for all $x,z\in [n]$ such that $ z \neq \bar 0$, as
    $$\ket{x,z}\mapsto \propto \ket{\psi_{z,x}(2^\ell)}$$
    can be implemented in $\tO(1)$ steps.
\end{lemma}
\begin{proof}
For a fixed $z,x \in [n], z \neq \bar 0$, we write down the state $\ket{\psi_{z,x}}$ by definition its definition (see \thm{basis_flows_circulations}) and split it into three orthogonal terms. (All sums below are over $\{0,1\}^{\log n}$).
\begin{align*}
    \ket{\psi_{z,x}(2^{\ell})} &= \sum_{j \in \{0,1\}^{\log n}} (-1)^{x \cdot j} \sum_{i \in \{0,1\}^{\log n}} (-1)^{z \cdot i} \ket{p_{ij}(2^\ell)}\\
    &= \sum_{j \in \{0,1\}^{\log n}} (-1)^{x \cdot j} \sum_{i \in \{0,1\}^{\log n}} (-1)^{z \cdot i} \left(\ket{\bar\theta^0_i(2^{\ell - 1})} + \ket{\bar\theta^{1i}_j(2^{\ell - 1})} - \ket{\bar\theta^{2j}_i(2^{\ell - 1})}\right)\\
    &= \sum_{j} (-1)^{x \cdot j} \sum_{i} (-1)^{z \cdot i} \ket{\bar\theta^0_i(2^{\ell - 1})}
    +\sum_{i,j} (-1)^{z \cdot i+x\cdot j} \ket{\bar\theta^{1i}_j(2^{\ell - 1})}
    - \sum_{i,j} (-1)^{z\cdot i+x \cdot j}\ket{\bar\theta^{2j}_i(2^{\ell - 1})}.
\end{align*}
By \eq{sign-explainer-function}, we can write $\ket{\bar\theta_i^{0}}=\ket{0,\bar{0}}\ket{\bar\theta_i}$, $\ket{\bar\theta_j^{1i}}=\ket{1,i}\ket{\bar\theta_j}$ and $\ket{\bar\theta_i^{2j}}=-\ket{2,j}\ket{\bar\theta_i}$.
If $x\neq \bar 0$, the first sum vanishes, otherwise it adds up to $n$, giving: 
\begin{align*}
    \ket{\psi_{z,x}(2^\ell)} & = \delta_{x,\bar{0}}n\cdot \ket{0,\bar{0}}\sum_{i \in \{0,1\}^{\log n}} (-1)^{z \cdot i} \ket{\bar\theta_i(2^{\ell - 1})}
    +\sum_{i \in \{0,1\}^{\log n}} (-1)^{z \cdot i} \ket{1,i}\sum_{j \in \{0,1\}^{\log n}} (-1)^{x \cdot j} \ket{\bar\theta_j(2^{\ell - 1})}\\
    &\qquad + \sum_{j \in \{0,1\}^{\log n}} (-1)^{x \cdot j}\ket{2,j} \sum_{i \in \{0,1\}^{\log n}} (-1)^{z \cdot i} \ket{\bar\theta_i(2^{\ell - 1})}\\
    & = \delta_{x,\bar{0}}n\ket{0,\bar{0}}\otimes \sqrt{N_z(2^{\ell - 1})}C_{2^{\ell - 1}}\ket{z}\\
    &\qquad +\ket{1}H^{\otimes\log n}\ket{z}\otimes \sqrt{N_x(2^{\ell - 1})}C_{2^{\ell - 1}}\ket{x}
    + \ket{2}H^{\otimes \log n}\ket{x}\otimes \sqrt{N_z(2^{\ell - 1})}C_{2^{\ell - 1}}\ket{z}
\end{align*}
where $C_{2^{\ell - 1}}$ is the circuit from \lem{Fourier_flows}, and $N_z(2^{\ell - 1})$ is as in \eq{N_xL}. Note by \eq{N_xL} that, for non-zero $z$, this is independent of $z$, so in that case, we simply write $N(2^{\ell - 1})$.

To prepare the superposition of two or three terms, we begin by creating a superposition of $\ket{0}$, $\ket{1}$ and $\ket{2}$ with appropriate amplitudes:
$$\ket{0}\ket{\bar{0}}\ket{x}\ket{z}\mapsto \left\{\begin{array}{ll}
  \frac{n\sqrt{N(2^{\ell - 1})}\ket{0}+\sqrt{N_{\bar{0}}(2^{\ell - 1})}\ket{1}+\sqrt{N(2^{\ell - 1})}\ket{2}}{\sqrt{(n^2+1)N(2^{\ell - 1})+N_{\bar{0}}(2^{\ell - 1})}}\ket{\bar{0}}\ket{x}\ket{z} & \mbox{if }x=\bar{0}\\
  \left(\frac{1}{\sqrt{2}}\ket{1}+\frac{1}{\sqrt{2}}\ket{2}\right)\ket{\bar{0}}\ket{x}\ket{z} & \mbox{else.}
\end{array}\right.$$
We have used the fact that when $x\neq \bar{0}$, $N_x(2^{\ell - 1})=N_z(2^{\ell - 1})$, so the amplitudes on $\ket{1}$ and $\ket{2}$ are equal. Using the closed-form expressions in \eq{sum_of_flows_norm} and \eq{N_xL}, we can compute these amplitudes in $\tO(1)$ steps, so we can implement this rotation in this number of steps as well. 

Next, conditioned on $\ket{1}$ in the first register, we swap the second and third registers. Then we apply a Hadamard to the second register conditioned on $\ket{1}$ or $\ket{2}$ in the first register, and apply $C_{2^{\ell - 1}}$ to the last register. \qedhere 
\end{proof}

Finally, if we take $\ell=\log L$, and $v_j=t$, in the following lemma, this gives a procedure for preparing a state proportional to the optimal unit flow $\ket{\theta_j(L)}$, which is the final state needed for the basis in \thm{basis_flows_circulations}.

\begin{lemma}\label{lem:flow_preparation}
    For all $\ell \in[\log(L)]$, a map that acts, for all $j\in \{0,1\}^{\log n}$, as
    $$\ket{j}\mapsto\propto \ket{\bar\theta_{j}(2^\ell)} = \frac{1}{n}\sum_{i\in \{0,1\}^{\log n}} \ket{p_{ij}(2^\ell)} $$
    can be implemented in $\tO(1)$ steps.
\end{lemma}

\begin{proof}
We prove the statement by induction. Let $\ell=1$ for the base case.
\begin{equation}\label{eq:last-base}
\begin{split}
    \ket{\bar\theta_j(2)} &= \frac{1}{n}\sum_{i\in\{0,1\}^{\log n}} \left(\ket{e^0_i} + \ket{e^{1i}_j} - \ket{e^{2j}_i}\right)\\
    &= \frac{1}{n}\sum_{i\in\{0,1\}^{\log n}} \ket{(0,\bar{0}),i} + \frac{1}{n}\sum_{i\in\{0,1\}^{\log n}} \ket{(1,i),j} + \frac{1}{n}\sum_{i\in\{0,1\}^{\log n}} \ket{(2,j),i}.
\end{split}
\end{equation}
Clearly, a state proportional to this can be prepared in $\tO(1)$ time.

For the induction step, assume there is an efficient circuit $C_{2^{\ell - 1}}'$ that implements the map
$$\ket{j}\mapsto \frac{\ket{\bar\theta_j(2^{\ell - 1})}}{\norm{\ket{\bar\theta_j(2^{\ell - 1})}}}$$
for any $j\in\{0,1\}^{\log n}$.
From \lem{optimal_flow_form} (see also \eq{cropped}) and the definition of $\ket{p_{ij}}$ in \thm{basis_flows_circulations}, we have:
 \begin{equation}\label{eq:theta_j-bar-L}
 \begin{split}
     &{n}\ket{\bar\theta_j(2^\ell)}
     =\sum_{i\in\{0,1\}^{\log n}} \ket{p_{ij}(2^{\ell})}\\
     ={}& \ket{0,\bar{0}}\sum_{i\in\{0,1\}^{\log n}} \ket{\bar\theta_i (2^{\ell - 1})} + \sum_{i\in\{0,1\}^{\log n}} \ket{1,i}\ket{\bar\theta_j (2^{\ell - 1})} + \ket{2,j}\sum_{i\in\{0,1\}^{\log n}} \ket{\bar\theta_i (2^{\ell - 1})}\\
     ={}& \ket{0,\bar{0}}\otimes \sqrt{N_{\bar{0}}(2^{\ell - 1})}C_{2^{\ell - 1}}\ket{\bar{0}} + \sqrt{n}\ket{1}H^{\otimes \log n}\ket{\bar{0}}\otimes \norm{\ket{\bar\theta_j (2^{\ell - 1})}}C_{2^{\ell - 1}}'\ket{j} + \ket{2,j}\otimes \sqrt{N_{\bar{0}}(2^{\ell - 1})}C_{2^{\ell - 1}}\ket{\bar{0}}
 \end{split}\end{equation}
 where $C_{2^{\ell - 1}}$ is the circuit from \lem{Fourier_flows}, and $N_{\bar{0}}(2^{\ell - 1})$ is as in \eq{sum_of_flows_norm}.

To prepare a state proportional to this, we first generate the superposition 
\begin{align*}
\ket{j}\mapsto & \frac{\sqrt{N_{\bar{0}}(2^{\ell - 1})}\ket{0}
+ \sqrt{n} \cdot \norm{\ket{\bar\theta_j(2^{\ell - 1})}} \ket{1}
+ \sqrt{N_{\bar{0}}(2^{\ell - 1})}\ket{2}}{\sqrt{n \cdot \norm{\ket{\bar\theta_j(2^{\ell - 1})}}^2 + 2 \cdot N_{\bar{0}}(2^{\ell - 1})}}\ket{j},
\end{align*}
before mapping $\ket{j}$ to the correct state, controlled on $\ket{0}$, $\ket{1}$, or $\ket{2}$, using swap, $C_{2^{\ell-1}}$, and $H^{\otimes \log n}$. 
By \lem{superposition_preparation}, such a superposition can be prepared in $\tO(1)$ time, provided that the amplitudes are computable in $\tO(1)$ time. As we already have an efficiently computable closed form expression for $N_{\bar{0}}$ (see \eq{sum_of_flows_norm}), it suffices to obtain a closed-form expression for $\norm{\ket{\bar\theta_j(2^{\ell - 1})}}$. For $\ell>0$, define $F_j(2^\ell) \coloneq \norm{\ket{\bar\theta_j(2^\ell)}}^2$. We will prove by induction that
\begin{equation}\label{eq:F-j}
    F_j(2^\ell)= \frac{1}{n^\ell}\cdot \frac{2(n+2)^\ell + n - 1}{n+1}.
\end{equation}
For the base step, let $\ell=1$. By \eq{last-base}, we have $F_j(2)=3/n$, which is equal to the right-hand side of \eq{F-j} for the setting $\ell=1$. Next, we show the induction step.
\begin{align*}
    F_j(2^\ell) &= \frac{1}{n^2} \left(2 N_{\bar{0}}(2^{\ell - 1})+n F_j(2^{\ell - 1})\right) & \mbox{by \eq{theta_j-bar-L}}\\
    &= \frac{2}{n} \left(1+\frac{2}{n}\right)^{\ell-1}+\frac{1}{n} F_j(2^{\ell - 1}) & \mbox{by \eq{sum_of_flows_norm}}\\
    &= \frac{2}{n^\ell}(n+2)^{\ell - 1} + \frac{1}{n^\ell}\cdot \frac{2(n+2)^{\ell - 1} + n - 1}{n+1} & \mbox{by i.h.}\\
    &= \frac{1}{n^\ell} \cdot \frac{2 (n + 1)(n+2)^{\ell - 1} + 2(n+2)^{\ell - 1} + n - 1}{n + 1}\\
    &= \frac{1}{n^\ell} \cdot \frac{2 (n + 2)(n+2)^{\ell - 1} + n - 1}{n + 1}
    = \frac{1}{n^\ell}\cdot \frac{2(n+2)^\ell + n - 1}{n+1}.
\end{align*}
\vskip-20pt
\end{proof}

\begin{corollary}\label{cor:cropped_flow_preparation}
    For all $\ell \in[\log(L)]$, the map that prepares a state proportional to
    $$\ket{\theta_{j}(2^\ell)} = -\ket{\leftarrow,[u]} + \frac{1}{n}\sum_{i\in \{0,1\}^{\log n}} \ket{p_{ij}(2^\ell)} +\ket{\rightarrow,[u,v_j]}$$
    can be implemented in $\tO(1)$ steps.
\end{corollary}

We summarize the results of the two sections in the following lemma, which states that the orthonormal basis of ${\cal B}^\perp$ from \thm{basis_flows_circulations} can be efficiently generated. 

\begin{lemma}\label{lem:basis generation}
    There exists an orthonormal basis of ${\cal B}^\perp ({\cal N}_{L}(s,t))$ that can be generated in time~$\tO (1)$.
\end{lemma}

\subsection{Complexity of the subroutine}\label{sec:final}

We conclude by combining the switching network and corresponding basis generation we have given in the previous sections with \thm{SN_algorithm} to prove \thm{quantum_short_path_complexity}, which states the existence of an algorithm for $\textsc{Dist}_L$.

\begin{proof}[Proof of \thm{quantum_short_path_complexity}]
From \lem{SN_correctness}, it follows that the switching network ${\cal N}_{L}(s,t)$ accepts $G$ if and only if there is a path of length at most $L$ from $s$ to $t$ in $G$. We apply \thm{SN_algorithm} to ${\cal N}_{L}(s,t)$. By \lem{positive_witness}, whenever there is a $st$-path of length $\leq L$ in $G$, there is a path from the source $[s]$ to the sink $[s,t]$ in ${\cal N}_{L}(s,t)(G)$ of length at most $W_+ := L^{\log 3}$. By \lem{negative_witness}, $|E({\cal N}_{L}(s,t))|\leq (2n+1)^{\log L} n$. By \lem{basis generation} and \lem{B_space_basis}, there exists an orthonormal basis of ${\cal B}^\perp ({\cal N}_{L}(s,t))$ that can be generated in time $T_B=\tO (1)$. 
Hence, by \thm{SN_algorithm}, there exists a quantum algorithm that decides whether there is a path of length at most $L$ from $s$ to $t$ in $G$ in time
$$O(T_B\sqrt{|E({\cal N}_L(u,v))| W_+}) = \tO \left( \left( L^{\log 3} (2n+1)^{\log L} n \right)^{1/2} \right)$$
and space
\begin{equation*}
    O(\log \abs{E({\cal N}_{L}(u,v))}) = O(\log L \log n).\qedhere
\end{equation*}
\end{proof}

\section*{Acknowledgments}

This work is co-funded by the European Union (ERC, ASC-Q, 101040624); the project Divide \& Quantum  (with project number1389.20.241) of the research programme NWA-ORC, which is (partly) financed by the Dutch Research Council (NWO); and the Dutch National Growth Fund (NGF), as part of the Quantum Delta NL programme. SJ is a CIFAR Fellow in the Quantum Information Science Program. 

\bibliographystyle{alpha}
\bibliography{Bibliography}

\end{document}